\documentclass{article}

\usepackage[dvips,dvipdfm,ps2pdf,hyperfigures,colorlinks]{hyperref}
\usepackage{amsmath,amssymb,graphicx,color,amsthm}
\usepackage{fullpage}

\newcommand{\prob}[1]{{\sf Pr}\ensuremath{\left[ #1 \right]}}
\newcommand{\probs}[2]{{\sf Pr}\ensuremath{#1}\ensuremath{\left[ #2 \right]}}
\newtheorem{theorem}{Theorem}[section]
\newtheorem{lemma}[theorem]{Lemma}
\newtheorem{definition}[theorem]{Definition}

\newcount\proofqeded
\newcount\proofended
\def\qed{ \mbox{\ \vrule width1ex height1em depth0cm}
\global\advance\proofqeded by 1 }

\def\eps{\epsilon}
\def\poly{\mathrm{poly}}
\def\polylog{\mathrm{polylog}}
\def\vbl{\mathrm{vbl}}

\def\P{\mathcal{P}}
\def\A{\mathcal{A}}
\def\ds{\displaystyle}
\def\Root{\Root}

\newcommand{\floor}[1]{\lfloor #1 \rfloor}

\newcommand{\comment}[1]{\bigskip \textcolor{red}{Comment: #1} \bigskip}

\newtheorem*{theorem-nonumber}{Theorem}
\newtheorem*{lemma-nonumber}{Lemma}

\begin{document}

\setcounter{page}{0}

\title{New Constructive Aspects of the Lov\'{a}sz Local Lemma\footnote{A preliminary version of this paper appeared in 51st Annual IEEE Symposium on Foundations of Computer Science (FOCS), 2010, Las Vegas.
}}

\author{Bernhard Haeupler\thanks{haeupler@mit.edu; CSAIL, Dept.\ of Computer
Science, Massachusetts Institute of Technology, Cambridge, MA 02139.
Part of this work was done while visiting the University of Maryland.}
\and Barna Saha\thanks{barna@cs.umd.edu; Dept.\ of Computer Science, University of Maryland, College Park, MD 20742
Supported in part by NSF Award CCF-0728839
and NSF Award CCF-0937865 and a Google Research Award} \and
Aravind Srinivasan\thanks{srin@cs.umd.edu; Dept.\ of Computer Science and Institute for
Advanced Computer Studies,
University of Maryland, College Park, MD 20742.
Supported in part by NSF ITR Award CNS-0426683, NSF Award CNS-0626636,
and NSF Award CNS 1010789.}
}

\date{\today}

\maketitle
\thispagestyle{empty}

\begin{abstract}
The Lov\'{a}sz Local Lemma (LLL) is a powerful tool that gives sufficient conditions for avoiding all of a given set of ``bad'' events, with positive probability. A series of results have provided algorithms to efficiently construct structures whose existence is non-constructively guaranteed by the LLL, culminating in the recent breakthrough of Moser \& Tardos for the full asymmetric LLL. We show that the output distribution of the Moser-Tardos algorithm well-approximates the \emph{conditional LLL-distribution} -- the distribution obtained by conditioning on all bad events being avoided. We show how a known bound on the probabilities of events in this distribution can be used for further probabilistic analysis and give new constructive and non-constructive results.

We also show that when a LLL application provides a small amount of slack, the number of resamplings of the Moser-Tardos algorithm is nearly linear in the number of underlying independent variables (not events!), and can thus be used to give efficient constructions in cases where the underlying proof applies the LLL to super-polynomially many events. Even in cases where finding a bad event that holds is computationally hard, we show that applying the algorithm to avoid a polynomial-sized ``core'' subset of bad events leads to a desired outcome with high probability. This is shown via a simple union bound over the probabilities of non-core events in the conditional LLL-distribution, and automatically leads to simple and efficient Monte-Carlo (and in most cases $RNC$) algorithms.

We demonstrate this idea on several applications. We give the first constant-factor approximation algorithm for the Santa Claus problem by making a LLL-based proof of Feige constructive. We provide Monte Carlo algorithms for acyclic edge coloring, non-repetitive graph colorings, and Ramsey-type graphs. In all these applications, the algorithm falls directly out of the non-constructive LLL-based proof. Our algorithms are very simple, often provide better bounds than previous algorithms, and are in several cases the first efficient algorithms known.

As a second type of application we show that the properties of the conditional LLL-distribution can be used in cases beyond the critical dependency threshold of the LLL: avoiding all bad events is impossible in these cases. As the first (even non-constructive) result of this kind, we show that by sampling a selected smaller core from the LLL-distribution, we can avoid a fraction of bad events that is higher than the expectation. MAX $k$-SAT is an illustrative example of this.
\end{abstract}

\newpage

\section{Introduction}
\label{sec:intro}
The well-known Lov\'{a}sz Local Lemma (LLL) \cite{erdos-lovasz:lll} is a powerful probabilistic approach to prove the existence of certain combinatorial structures. Its diverse range of applications include breakthroughs in packet-routing \cite{lmr}, a variety of theorems in graph-coloring including list coloring, frugal coloring, total coloring, and coloring graphs with lower-bounded girth \cite{molloy-reed:book}, as well as a host of other applications where probability appears at first sight to have no role \cite{alon-spencer}. Furthermore, almost all known applications of the LLL have no alternative proofs known. While the original LLL was non-constructive -- it was unclear how the existence proofs could be turned into polynomial-time algorithms -- a series of works \cite{beck:lll, alon:lll-journal, czumaj:lll, molloy:lll, molloy-reed:book, srin:lll, moser:derand, moser:lll, moser-tardos:lll} beginning with Beck \cite{beck:lll} and culminating with the breakthrough of Moser \& Tardos (MT) \cite{moser-tardos:lll} have led to efficient algorithmic versions for most such proofs. However, there are several LLL applications to which these approaches inherently cannot apply; our work makes progress toward bridging this gap, by uncovering and exploiting new properties of \cite{moser-tardos:lll}. We also obtain what are, to our knowledge, the first algorithmic applications of the LLL where a few of the bad events have to happen, and
where we aim to keep the number of these small.

%AS
We will use standard notation: $e$ denotes the base of the natural logarithm,
and $\ln$ and $\log$ denote the logarithm to the base $e$ and $2$,
respectively.

Essentially all known applications of the LLL use the following framework. Let $\mathcal{P}$ be a collection of $n$ mutually independent random variables $\{P_1, P_2, \ldots,P_n\}$, and let $\A=\{A_1, A_2, \ldots, A_m\}$ be a collection of $m$ (``bad'') events, each determined by some subset of $\mathcal{P}$. The LLL (Theorem~\ref{theorem:lll}) shows sufficient conditions under which, with positive probability, none of the events $A_i$ holds: i.e., that there is a choice of values for the variables in $\mathcal{P}$ (corresponding to a discrete structure such a suitable coloring of a given graph) that avoids all the $A_i$. Under these same sufficient conditions, MT shows the following very simple algorithm to make such a choice: (i) initially choose the $P_i$ independently from their given distributions; (ii) \textit{while} the current assignment to $\mathcal{P}$ does not avoid all the $A_i$, \textit{repeat}: arbitrarily choose a currently-true $A_i$, and resample, from their product distribution, the variables in $\mathcal{P}$ on which $A_i$ depends. The amazing aspect of MT is that the expected number of resamplings is small \cite{moser-tardos:lll}: at most $\poly(n,m)$ in all known cases of interest. However, there are two problems with implementing MT, that come up in some applications of the LLL:
\begin{description}
\item[(a)] the number of events $m$ can be superpolynomial in the number of variables $n$; this can result in a superpolynomial running time in the ``natural'' parameter $n$ \footnote{$n$ is the parameter of interest since the output we seek is one value for each of $P_1, P_2, \ldots, P_n$.}; and, even more seriously,
\item[(b)] given an assignment to $\mathcal{P}$, it can be computationally hard (e.g., NP-hard or yet-unknown to be in polynomial time) to either certify that no $A_i$ holds, or to output an index $i$ such that $A_i$ holds.
\end{description}

Since detection and resampling of a currently-bad event is the seemingly unavoidable basic step in the MT algorithm, these applications seemed far out of reach.
We deal with a variety of applications wherein (a) and/or (b) hold, and develop Monte Carlo (and in many cases, $RNC$) algorithms whose running time is polynomial in $n$: some of these applications involve a small loss in the quality of the solution. (We loosely let ``$RNC$ algorithms'' denote randomized parallel algorithms that use $\poly(n)$ processors and run in $\polylog(n)$ time, to output a correct solution with high probability.) First we show that the MT algorithm needs only $O(n^2 \log n)$ many resampling steps in all applications that are known (and in most cases $O(n \cdot \polylog(n))$), even when $m$ is superpolynomial in $n$. This makes those applications constructive that allow an \emph{efficient} implicit representation of the bad events (in very rough analogy with the usage of the ellipsoid algorithm for convex programs with exponentially many constraints but with good separation oracles). Still, most of our applications have problem (b). For these cases, we introduce a new proof-concept based on the \emph{(conditional) LLL-distribution} -- the distribution $D$ on $\mathcal{P}$ that one obtains when conditioning on no $A_i$ happening. Some very useful properties are known for $D$ \cite{alon-spencer}: informally, if $B$ depends ``not too heavily'' on the events in $\A$, then the probability placed on $B$ by $D$ is ``not much more than'' the unconditional probability $\prob{B}$: at most $f_{\mathcal{A}}(B) \cdot \prob{B}$ (see (\ref{eqn:cond-lll})). Such bounds in combination with further probabilistic analysis can be used to give interesting (nonconstructive) results. Our next main contribution is that the MT algorithm has an output distribution (say $D'$) that ``approximates'' the LLL-distribution $D$: in that for every $B$, the \emph{same} upper bound $f_{\mathcal{A}}(B) \cdot \prob{B}$ as above, holds in $D'$ as well. This can be used to make probabilistic proofs that use the LLL-condition constructive.
%We give two examples of this next.

Problem (b), in all cases known to us, comes from problem (a): it is easy
to test if any \emph{given} $A_i$ holds currently (e.g., if a given subset
of vertices in a graph is a clique), with the superpolynomiality of $m$
being the apparent bottleneck. To circumvent this, we develop our third
main contribution: the very general Theorem~\ref{thm:polycoremain}
that is simple and directly applicable in all LLL instances that allow a small slack
%(``exponential $\eps$-slack'')
in the LLL's sufficient conditions. This theorem proves that a small
$\poly(n)$-sized core-subset of the events in $\A$ can be selected and
avoided efficiently using the MT algorithm. Using the LLL-distribution
and a simple union bound over the non-core events,
we get efficient (Monte Carlo and/or $RNC$) algorithms for these problems.

We develop two types of applications, as sketched next.

\subsection{Applications that avoid all bad events}
\label{sec:app-all-avoided}
A summary of four applications follows; all of these have problem (a), and all but the acyclic-coloring application have problem (b). Most such results have $RNC$ versions as well.

\smallskip \noindent
\emph{The Santa Claus Problem:}
The Santa Claus problem is the restricted assignment version of the max-min allocation problem of indivisible goods. The Santa Claus has $n$ items that need to be distributed among $m$ children. Each child has a utility for each item, which is either $0$ or some given $p_{j}$ for item $j$. The objective is to assign each item to some child, so that the minimum total utility received by any child is maximized. This problem has received much attention recently \cite{bansal:stoc06,asadpour:stoc07,feige:soda08,asadpour-feige-saberi,bateni:stoc09,julia:focs09}. The problem is NP-Hard and the best-known approximation algorithm due to Bansal and Sviridenko \cite{bansal:stoc06} achieves an approximation factor of $O(\frac{\log{\log{m}}}{\log{\log{\log{m}}}})$ by rounding a certain configuration LP. Later, Feige in \cite{feige:soda08} and subsequently Asadpour, Feige and Saberi in \cite{asadpour-feige-saberi} showed that the integrality gap of the configuration LP is a constant. Surprisingly, both results were obtained using two different non-constructive approaches and left the question for a constant-factor approximation algorithm open. This made the Santa Claus problem to one of the rare instances \cite{Feige08-estimationvsapproximation} in which the proof of an integrality gap did not result in an approximation algorithm with the same ratio. In this paper we resolve this by making the non-constructive LLL-based proof of Feige \cite{feige:soda08} constructive (Section \ref{sec:santa}) and giving the first constant-factor approximation algorithm for the Santa Claus problem.

\smallskip \noindent
\emph{Non-repetitive Coloring of Graphs:}
Given a  graph $H=(V,E)$, a $k$-coloring (not necessarily proper) of the edges of $H$ is called \emph{non-repetitive} if the sequence of colors along any simple path is not the same in the first and the second half. The smallest $k$ such that $H$ has a non-repetitive $k$-coloring is called the {\em Thue number} $\pi(H)$ of $H$ \cite{thue:1906}. Alon, Grytczuk, Hauszczak and Riordan showed via the LLL that $\pi(H)\leq O(\Delta(H)^{2})$ \cite{alon:random02}, where $\Delta$ is the maximum degree of any vertex in $H$. This was followed by much additional works \cite{james:thr05, marcus:sigact02, jaroslaw:discrete, kundgen:08, bresar:07, alonG:08}. However, no efficient construction is known till date, except for special classes of graphs such as complete graphs, cycles and trees. We present a randomized algorithm for non-repetitive coloring of $H$ using at most $O(\Delta(H)^{2+\epsilon})$ colors, for every constant $\epsilon > 0$ (Section \ref{sec:non-rep}).

\smallskip \noindent
\emph{General Ramsey-Type Graphs:}
The Ramsey number $R(U_s,V_t)$ refers to the
smallest $n$ such that any graph on $n$ vertices either contains a
 $U_{s}$ within any subgraph of $s$ vertices, or there exist $t$ vertices
 that do not contain $V_{t}$. Obtaining lower bounds for various special
 cases of $R(U_s,V_t)$ and constructing Ramsey type graphs have been studied in
 much detail \cite{alon94,alon99,kriv2,alonkriv}.
% Constructions (probabilistic) have been given for the cases when $U_s=K^s$ and
% $V_t=K^2$ (Ramsey number), or $V_{t}=K^{r}$ and for general $U_s$ \cite{alonkriv,kriv2}.
A predominant case for such problems is when $s$ is held fixed.
% Alon et al. gave a deterministic construction for Ramsey graph ($R(K^s,K^2)$)
% with fixed $s$. This is referred to as off-diagonal Ramsey number \cite{alon94}.
 We consider the general setting of $R(U_s,V_t)$ with fixed $s$, and provide
efficient randomized algorithms for constructing Ramsey-type graphs (Section \ref{sec:ramsey}).

\smallskip \noindent
\emph{Acyclic Edge-Coloring:}
%An edge-coloring of a graph is proper if no pair of incident edges are colored the same.
A proper edge-coloring of a graph is \emph{acyclic} iff each cycle in it receives more than $2$ colors. The acyclic chromatic number $a(G)$ introduced in \cite{baum:acyclic} is the minimum number of colors in a proper acyclic edge coloring of $G$ \cite{alon:acyclic,molloy:lll,alon:acyclic1,baum:acyclic,muthua:acyclic}. Alon, McDiarmid and Reed \cite{alon:acyclic} showed that $a(G) < 64\Delta$, where $\Delta$ is the maximum degree. The constant was later improved to $16$ by Molloy and Reed \cite{molloy:lll}, who also mention an algorithmic version using $20\Delta$ colors.
However it was conjectured that $a(G)=\Delta +2$; Alon, Sudakov and Zaks showed indeed the conjecture is true for graphs having girth $\Omega(\Delta \log{\Delta})$ \cite{alon:acyclic1}. Their algorithm can be made constructive using Beck's technique \cite{beck:lll} to obtain an acyclic edge coloring using $\Delta +2$ colors, albeit for graphs with girth significantly larger than $\Theta(\Delta \log{\Delta})$ \cite{alon:acyclic1}. We bridge this gap by providing constructions to achieve the same girth bound as in \cite{alon:acyclic1}, yet obtaining an acyclic edge coloring with only $\Delta+2$ colors. For graphs with no girth bound, $16\Delta$ colors suffice to efficiently construct an acyclic edge coloring in contrast to the $20\Delta$ algorithmic bound of \cite{molloy:lll} (Section \ref{sec:acyclic}).

The recent result of Matthew Andrews on approximating the edge-disjoint
paths problem on undirected graphs is another example, where problems
(a) and (b) occur and our LLL-techniques are applied to avoid
super-polynomially many bad events  \cite{andrews:focs10}.

\subsection{Applications that avoid many bad events}
\label{sec:app-many-avoided}
Many settings require ``almost all'' bad events to be avoided, and not
necessarily all; e.g., consider MAX-SAT as opposed to SAT. However, in the
LLL context, essentially the only known general applications were ``all or
nothing'': either the LLL's sufficient conditions hold, and we are able to
avoid all bad events, or the LLL's sufficient conditions are violated, and
the only known bound on the number of bad events is the trivial one given
by the linearity of expectation (which does not exploit any
``almost-independence'' of the bad events, as does the LLL). This situation is
even more pronounced in the algorithmic setting. We take what are, to our
knowledge, the first steps in this direction, interpolating between these
two extremes.

While our discussion here holds for all applications of the symmetric LLL,
let us take MAX-$k$-SAT as an illustrative example.
(The LLL is stated in Section~\ref{sec:algframework}, but let us
recall its well-known ``symmetric'' special case: in the setting of MT with
$\P$ and $A$ as defined near the beginning of Section \ref{sec:intro},
if $\prob{A_i} \leq p$ and $A_i$ depends on at most $d$ other $A_j$ for all
$i$, then $e \cdot p \cdot (d + 1) \leq 1$ suffices to avoid
all the $A_i$.)
Recall that in MAX-$k$-SAT, we have a CNF formula
on $n$ variables, with $m$ clauses each containing exactly $k$ literals; as
opposed to SAT, where we have to satisfy all clauses, we aim to maximize the
number of satisfied clauses here. The best general upper-bounds on the
number of ``violated events'' (unsatisfied clauses) follow from the
probabilistic method, where each variable is set to True or False uniformly
at random and independently. On the one hand, the linearity of expectation
yields that the expected number of unsatisfied clauses is $m \cdot 2^{-k}$
(with a derandomization using the method of conditional probabilities). On the
other hand, if each clause shares a variable with at most $2^k/e - 1$ other
clauses, a simple application of the symmetric LLL shows that all clauses
can be satisfied (and made constructive using MT). No interpolation
between these was known before; among other results, we show that if each
clause shares a variable with at most $\sim \alpha 2^k/e$ other clauses
for $1 < \alpha < e$, then
we can efficiently construct an assignment to the variables that violates
at most $(e \ln(\alpha) / \alpha + o(1)) \cdot m \cdot 2^{-k}$ clauses for
large $k$.
(This is better than the linearity of expectation iff $\alpha < e$: it is
easy to construct examples with $\alpha = e$ where one cannot do better
than the linearity of expectation. See \cite{agksy:max-sat} for the
fixed-parameter tractability of MAX-$k$-SAT above
$(1 - 2^{-k}) m$ satisfied clauses.)

The above and related results for applications of the symmetric LLL,
follow from the connection to the ``further probabilistic
analysis using the remaining randomness of LLL-distributions'' that we alluded
to above; see Section \ref{sec:beyond}.
We believe this connection to be the main conceptual message of
this paper, and expect further applications in the future.

\subsection{Preliminaries \& Algorithmic Framework}
\label{sec:algframework}
We follow the general algorithmic framework of the Local Lemma due to MT.
As in our description at the beginning of Section \ref{sec:intro}, let $\P$ be a finite collection of mutually independent
random variables $\{P_1, P_2, \ldots,P_n\}$ and let $\A=\{A_1, A_2, \ldots, A_m\}$ be a collection of events, each determined by some subset of $\P$. For any event $B$ that is determined by a subset of $\P$ we denote the smallest such subset by $\vbl(B)$.
For any event $B$ that is determined by the variables in $\P$, we furthermore write $\Gamma(B) = \Gamma_{\A}(B)$ for the set of all events $A \neq B$ in $\A$ with $\vbl(A) \cap \vbl(B) \neq \emptyset$. This neighborhood relation induces the following standard \emph{dependency graph} or \emph{variable-sharing graph} on $\A$: For the vertex set $\A$ let $G = G_{\A}$ be the undirected graph with an edge between events $A, B \in \A$ iff $A \in \Gamma(B)$. We often refer to events in $\A$ as \emph{bad events} and want to find a point in the probability space, or equivalently an assignment to the variables $\P$, wherein none of the bad events happen. We call such an assignment a \textbf{good assignment}.

With these definitions the general (``asymmetric'') version of the LLL simply states:

\begin{theorem}[Asymmetric Lov\'{a}sz Local Lemma]
\label{theorem:lll}
With $\A,\P$ and $\Gamma$ defined as above, if there exists an assignment of reals $x: \mathcal{A} \rightarrow (0,1)$ such that
\begin{equation}
\label{eqn:alll}
\forall A\in \A: \prob{A}\leq x(A) \prod_{B \in \Gamma(A)}(1-x(B));
\end{equation}
then the probability of avoiding all bad events is at least $\Pi_{A \in \mathcal{A}}(1-x(A)) > 0$ and thus there exists a good assignment to the variables in $\P$.
\end{theorem}

We study several LLL instances where the number of events to be avoided, $m$, is super-polynomial in $n$; our goal is to develop algorithms whose running time is polynomial in $n$ which is also the size of the output - namely a good assignment of values to the $n$ variables. We introduce a key parameter:
\begin{equation}
\label{eqn:defn-delta}
\delta := \min_{A \in \A} x(A) \prod_{B \in \Gamma(A)}(1-x(B)).
\end{equation}
Note that without loss of generality $\delta \leq \frac{1}{4}$ because otherwise all $A \in \A$ are independent, i.e., defined on disjoint sets of variables. Indeed if $\delta > \frac{1}{4}$ and there is an edge in $G$ between $A \in \A$ and $B \in \A$ then we have $\frac{1}{4} > x(A) (1 - x(B))$ and $\frac{1}{4} > x(B) (1 - x(A))$, i.e., $\frac{1}{4} \cdot \frac{1}{4} > x(A) (1 - x(A)) \cdot x(B) (1 - x(B))$ which is a contradiction because $x(1-x)\leq \frac{1}{4}$ for all $x$ (the maximum is attained at $x=\frac{1}{2}$).

We allow our algorithms to have a running-time that is polynomial in $\log(1/\delta)$; in all applications known to us, $\delta \geq \exp(-O(n \log n))$, and hence, $\log(1/\delta) = O(n \log n)$. In fact because $\delta$ is an upper bound for $\min_{A \in \A} P(A)$ in any typical encodings of the domains and the probabilities of the variables, $\log(1/\delta)$ will be at most linear in the size of the input or the output.

The following subsection \ref{sec:MTreview} reviews the MT algorithm and its
analysis, which will be helpful to understand some of our proofs and
technical contributions; the reader familiar with the MT algorithm may skip it.

\subsection{Review of the MT Algorithm and its Analysis}
\label{sec:MTreview}

Recall the resampling-based MT algorithm; let us now review some of the technical elements in the analysis of this algorithm,
that will help in understanding our technical contributions better.

A witness tree $\tau=(T,\sigma_{T})$ is a finite rooted tree $T$ together with a labeling
$\sigma_{T}:V(T) \rightarrow \mathcal{A}$ of its vertices to events,
such that the children
of a vertex $u \in V(T)$ receive labels from $\Gamma(\sigma_{T}(u))\cup \sigma_{T}(u)$.
In a proper witness tree distinct children of the
same vertex always receive distinct labels. The ``log'' $C$ of an
execution of MT lists the events as they
have been selected for resampling in each step. Given $C$, we can associate a witness tree
$\tau_{C}(t)$ with each resampling step $t$ that can serve as a justification for the
necessity of that correction step. $\tau_{C}(t)$ will be rooted at $C(t)$. A witness tree is
said to occur in $C$, if there exists $t \in N$, such that $\tau_{C}(t)=\tau$. It has been shown
in \cite{moser-tardos:lll} that if $\tau$ appears in $C$, then
it is proper and it appears in $C$ with probability at most
$\Pi_{v \in V(\tau)}\prob{\sigma_{T}(v)}$.

To bound the running time of the MT algorithm, one needs to bound the number of times an event $A \in \A$ is
resampled. If $N_{A}$ denotes the random variable for the number of resampling steps of $A$ and
$C$ is the execution log; then $N_{A}$ is the number of occurrences of $A$ in this log and also
the number of distinct proper witness trees occurring in $C$ that have their root labeled $A$. As a
result one can bound the expected value of $N_{A}$ simply by summing the probabilities of appearances of
distinct witness trees rooted at $A$. These probabilities can be related to a Galton-Watson branching process
to obtain the desired bound on the running time.

A Galton-Watson branching process can be used to generate a proper witness tree as follows. In the first
round the root of the witness tree is produced, say it corresponds to event $A$. Then in each subsequent round, for
each vertex $v$ independently and again independently, for each event $B \in \Gamma{\sigma_{T}(v)}\cup\sigma_{T}(v)$,
$B$ is selected as a child of $v$ with probability $x(B)$ and is skipped with probability $(1-x_{B})$. We will use
the concept of a proper witness trees and Galton-Watson process in several of our proofs.

\section{LLL-Distribution}
\label{sec:llldistrib}
When trying to turn the non-constructive Lov\'{a}sz Local Lemma into an algorithm that finds a good assignment the following straightforward approach comes to mind: draw a random sample for the variables in $\P$ until one is found that avoids all bad events. If the LLL-conditions are met this rejection-sampling algorithm certainly always terminates but because the probability of obtaining a good assignment is typically exponentially small it takes an expected exponential number of resamplings and is therefore non-efficient. While the celebrated algorithm of Moser (and Tardos) is much more efficient, the above rejection-sampling method has a major advantage: it does not just produce an arbitrary assignment but provides a randomly chosen assignment from the distribution
that is obtained when one conditions on no bad event happening. In the following, we call this distribution \emph{LLL-distribution} or \emph{conditional LLL-distribution}.

The LLL-conditions and further probabilistic analysis can be a powerful tool
to obtain new results (constructive or otherwise) like the constructive one
in Section \ref{sec:beyond}. The following is a well-known bound on the
probability $\probs{_D}{B}$ that the LLL-distribution $D$ places on \emph{any}
event $B$ that is determined by variables in $\P$ (its proof is an easy
extension of the standard non-constructive LLL-proof \cite{alon-spencer}):

\begin{theorem}\label{thm:llldistrib-nonconstr}
If the LLL-conditions from Theorem \ref{theorem:lll} are met, then the LLL-distribution $D$ is well-defined. For any
event $B$ that is determined by $\P$, the probability \probs{_D}{B} of $B$ under $D$ satisfies:
\begin{equation}
\label{eqn:cond-lll}
\probs{_D}{B} := \prob{B \bigm| \bigwedge_{A \in \A} \overline{A}} \leq
\prob{B} \cdot \prod_{C \in \Gamma(B)} (1 - x_C)^{-1};
\end{equation}
here, $\prob{B}$ is the probability of $B$ holding under a random choice of $P_1, P_2, \ldots, P_n$.
\end{theorem}

The fact that the probability of an event $B$ does not increase much in the conditional LLL-distribution when $B$ does not depend on ``too many'' $C \in \A$, is used critically in the rest of the paper.

More importantly, the following theorem states that the output distribution $D'$ of the MT-algorithm
approximates the LLL-distribution $D$ and has the very nice property that it essentially also satisfies (\ref{eqn:cond-lll}):

\begin{theorem}
\label{thm:distrib}
Suppose there is an assignment of reals $x:\mathcal{A}\rightarrow (0,1)$ such that (\ref{eqn:alll}) holds. Let $B$ be any event that is determined by $\P$. Then, the probability that $B$ was true \emph{at least once}
during the execution of the MT algorithm on the events in $\mathcal{A}$, is at most $\prob{B} \cdot (\prod_{C \in \Gamma(B)} (1 - x_C))^{-1}$. In particular the probability of $B$ being true in the output distribution of MT obeys this upper-bound.
\end{theorem}
\begin{proof}
The bound on the probability of $B$ ever happening is a simple extension of the MT proof \cite{moser-tardos:lll}. Note that we want to prove the theorem irrespective of whether $B$ is in $\A$ or not. In either case we are interested in the probability that the event was true at least once during the execution, i.e., if $B$ is in $\A$ whether it could have been resampled at least once. The witness trees that certify the \emph{first time} $B$ becomes true are the ones that have $B$ as a root and all non-root nodes from $\A \setminus \{B\}$. Similarly as in \cite{moser-tardos:lll}, we calculate the expected number of these witness trees via a union bound. Let $\tau$ be a fixed proper witness tree with its root vertex labeled $B$. Following the proof of Lemma 3.1 and using the fact that $B$ cannot be a child of itself, it can be shown that the probability $p_{\tau}$ with which the Galton-Watson process that starts with $B$ yields exactly the tree $\tau$ is $p_{\tau}=\prod_{A \in \Gamma(B)} (1 - x(A)) \cdot \prod_{v \in V(\tau)}x'(\sigma_{v})$. Here $V(\tau)$ are the non-root vertices of $\tau$ and $x'(\sigma_{v})=x(\sigma_{v})\prod_{C \in \Gamma(\sigma_{v})}(1-x(C))$. Plugging this in the arguments following the proof of Lemma 3.1 of \cite{moser-tardos:lll} it is easy to see that the union bound over all these trees and therefore also the desired probability is at most $\prob{B} \cdot (\prod_{C \in \Gamma(B)} (1 - x_C))^{-1}$ where the term ``\prob{B}'' accounts for the fact that the root-event $B$ has to be true as well.
\end{proof}

Using this theorem we can view the MT algorithm as an \emph{efficient} way to obtain a sample
that comes approximately from the conditional LLL-distribution. This efficient sampling procedure makes it possible to make proofs using the conditional LLL-distribution constructive and directly convert them into algorithms. All constructive results of this paper are based on Theorem \ref{thm:distrib} and demonstrate this idea. 
\section{LLL Applications with Super-Polynomially Many Bad Events}
\label{sec:superpoly}
In several applications of the LLL, the number of bad events is super-polynomially larger than the number of underlying variables. In these cases we aim for an
algorithm that still runs in time polynomial
in the number of variables, and it is not efficient to have an explicit representation of all bad events.
Surprisingly, Theorem \ref{thm:runningtime} shows that the number of resamplings done by the MT algorithm
remains quadratic and in most cases even near-linear in the number of variables $n$.

\begin{theorem}\label{thm:runningtime}
Suppose there is an $\eps \in [0,1)$ and an assignment of reals $x:\mathcal{A}\rightarrow (0,1)$ such that:
$$\forall A\in \A: \prob{A}\leq (1 - \eps) x(A) \prod_{B \in \Gamma(A)}(1-x(B)).$$
With $\delta$ denoting $\min_{A \in \A} x(A) \prod_{B \in \Gamma(A)}(1-x(B))$,
we have
\begin{equation}
\label{eqn:T-bound}
T := \sum_{A \in \A} x_A \leq n \log(1/\delta).
\end{equation}
Furthermore:
\begin{enumerate}
\item if $\eps=0$, then the expected number of resamplings done by the MT algorithm is at most
$v_1 = T \max_{A \in \A} \frac{1}{1 - x(A)}$, and for any parameter $\lambda \geq 1$, the MT algorithm terminates within $\lambda v_1$ resamplings with probability at least $1 - 1/\lambda$.
\item if $\eps > 0$, then the expected number of resamplings done by the MT algorithm is at most
$v_2 = O(\frac{n}{\eps} \log \frac{T}{\eps})$, and for any parameter $\lambda \geq 1$, the MT algorithm terminates within $\lambda v_2$ resamplings with probability $1 - \exp(-\lambda)$.
\end{enumerate}
\end{theorem}
\begin{proof}
The main idea of relating the quantity $T$ to $n$ and $\delta$ is to use:
(i) the fact that the variable-sharing graph $G$ is very dense, and
(ii) the nature of the LLL-conditions which force highly connected events to have small probabilities and $x$-values. To see that $G$ is dense,
consider for any variable $P \in \P$ the set of events
\[ \A_P = \{A \in \A | P \in \vbl(A)\}, \]
and note that these events form a clique in $G$. Indeed, the $m$ vertices
of $G$ can be partitioned into $n$ such cliques with potentially further
edges between them, and therefore has at least
$n \cdot {{m/n} \choose 2} = m^2/(2n) - m/2$ edges, which is high
density for $m \gg n$.

Let us first prove the bound on $T$. To do so, we fix
any $P \in \P$ and show that
$\sum_{B \in \A_P} x_B \leq \log(1/\delta)$,
which will clearly suffice.
Recall from the discussion following (\ref{eqn:defn-delta}) that we can assume w.l.o.g.\ that
$\delta \leq \frac{1}{4}$.
If $|A_P| = 1$, then of course $\sum_{B \in \A_P} x_B \leq 1 \leq
\log(1/\delta) $.
If $|A_P| > 1$, let $A \in \A_P$ have the smallest $x_A$ value.
Note that by definition
\[ \delta \leq x_A \prod_{B \in \A_P\setminus A} (1 - x_B)=\frac{x_A}{1-x_A} \prod_{B \in \A_P} (1 - x_B). \]
If $x_A \leq 1/2$, then $\delta \leq \prod_{B \in \A_P} (1 - x_B) \leq e^{-\sum_{B \in \A_P} x_B}$,
and we get $\sum_{B \in \A_P} x_B \leq \ln{(1/\delta)} < \log(1/\delta)$
as required.
Otherwise, if $x_A > 1/2$, let $B_1 \in \A_{P} \setminus A$. Then,
\begin{equation}
\label{eqn:T-bound1}
\delta \leq x_A \cdot \prod_{B \in \A_P \setminus A} (1 - x_B)
= x_A(1-x_{B_1}) \prod_{B \in \A_P \setminus (A \cup B_1)} (1 - x_B)
\leq x_A(1-x_{B_1}) e^{-\sum_{B \in \A_P \setminus (A \cup B_1)} x_B}.
\end{equation}
Let us now show that for $1/2 \leq x_A \leq x_{B_1} \leq 1$,
\begin{equation}
\label{eqn:T-bound2}
x_A(1-x_{B_1}) \leq e^{-(x_A + x_{B_1})}.
\end{equation}
Fix $x_{A}$. We thus need to show $e^{x_{B_1}}(1-x_{B_1})\leq \frac{1}{x_{A}e^{x_{A}}}$.
The derivative of $e^{x_{B_1}}(1-x_{B_1})$ is negative for $x_{B_1}\geq 0$, showing that it is a decreasing
function in the range $x_{B_1} \in [x_A,1]$. Therefore the maximum value of $e^{x_{B_1}}(1-x_{B_1})$ is obtained
at $x_{B_1}=x_{A}$ and for (\ref{eqn:T-bound2}) to hold, it is enough to show that,
$x_A(1-x_{A}) \leq e^{-2x_A }$ holds. The second derivative of $e^{-2x_{A}}-x_{A}(1-x_{A})$
is positive. Differentiating $e^{-2x_{A}}-x_{A}(1-x_{A})$ and equating the derivative to $0$,
returns the minimum in $[1/2,1]$ at $x_{A}=0.7315$. The minimum value is
$0.0351 > 0$.
Thus we have (\ref{eqn:T-bound2}) and so we get
\[x_A(1-x_{B_1}) e^{-\sum_{B \in \A_P \setminus (A \cup B_1)} x_B} \leq e^{-\sum_{B \in \A_P} x_B}; \]
using this with (\ref{eqn:T-bound1}), we obtain
$\sum_{B \in \A_P} x_B \leq \ln{(1/\delta)} < \log(1/\delta)$ as desired.
%since $x_B > 1/2$ for all $B \in A_P$ and therefore also $\sum_{B \in \A_P} x_B = O(\log(1/\delta))$. Summing over all variables gives $T = \sum_{A \in \A} x_A = O(n \log(1/\delta))$.

Given the bound on $T$, part (1) follows directly from the main theorem of \cite{moser-tardos:lll} and by a simple application of Markov's inequality.

Part (2) now also follows from \cite{moser-tardos:lll}. In section 5 of \cite{moser-tardos:lll} it is shown that saving a $(1-\eps)$ factor in the probability of every resampling step implies that with high probability, no witness tree of size $\Omega(\frac{1}{\eps} \log \sum_{A \in \A} \frac{x_A}{1-x_A})$ occurs. This easily implies that none of the $n$ variables can be resampled more often. It is furthermore shown that without loss of generality all $x$-values can be assumed to be bounded away from $1$ by at least $O(\eps)$. This simplifies the upper bound on the expected running time to $n \cdot O(\frac{1}{\eps} \log \frac{T}{\eps})$.
\end{proof}

As mentioned following the introduction of $\delta$ in (\ref{eqn:defn-delta}),
$\log(1/\delta) \leq O(n \log n)$ in all applications known to us, and is
often even smaller.

\paragraph{Remarks}
\begin{itemize}
\item
The $\ds \max_{A \in \A} \frac{1}{1 - x(A)}$ factor in the running time of part (1) of Theorem \ref{thm:runningtime} corresponds to the expected number of times the event $A$ gets resampled until one satisfying assignment to its variables is found. It is obviously unavoidable for an algorithm that has only black-box resampling and evaluation access to the events. If one alters the algorithm to pick a random assignment that satisfies $A$ (which can for example be computed using rejection sampling, taking an expected $\Theta(\frac{1}{1 - x(A)})$ trials each time), this factor can be avoided.
\item The estimation $T = \sum_{A \in \A} x_A = O(n \log 1/\delta)$ is tight and can be achieved, e.g., by having an isolated event with constant probability for each variable. In many cases with $\log 1/\delta = \omega(\log n)$ it is nevertheless an overestimate, and in most cases the running time is
$O(n \log n)$ even for $\eps = 0$.
\end{itemize}

While Theorem~\ref{thm:runningtime} gives very good bounds on the running
time of MT even for applications with $\Omega(n) \leq m
\leq \poly(n)$ many events, it unfortunately often fails to be directly
applicable when $m$ becomes super-polynomial in $n$. The reason is that maintaining bad events implicitly and running the resampling process requires an efficient way to find violated events. In many examples like those of Section \ref{sec:santa}, \ref{sec:non-rep} and \ref{sec:ramsey} with super-polynomially many events, finding violated events or even just verifying a good assignment is not known to be in polynomial time (often even provably NP-hard). To capture the sets of events for which we can run the MT algorithm efficiently we use the following definition:

\begin{definition}
\label{defn:verifiable}
\textbf{(Efficient verifiability)}
A set $\A$ of events that are determined by variables in $\P$ is \emph{efficiently verifiable} if, given an arbitrary assignment to $\P$, we can efficiently find an event $A \in \A$ that holds or detect that there is no such event.
\end{definition}

Because many large $\A$ of interest are not efficiently verifiable, a direct application of the MT-algorithm is not efficient. Nevertheless we show in the rest of this section that using the randomness in the output distribution of the MT-algorithm characterized by Theorem \ref{thm:distrib}, it is still practically always possible to obtain efficient Monte Carlo algorithms that produce a good assignment with high probability.

The main idea is to judiciously select an efficiently verifiable \emph{core subset} $\A' \subseteq \A$ of bad events and apply the MT-algorithm to it. Essentially instead of looking for violated events in $\A$ we only resample events from $\A'$ and terminate when we cannot find one such violated event. The non-core events will have small probabilities and will be sparsely connected to core events and as such their probabilities in the LLL-distribution and therefore also the output distribution of the algorithm does not blow up by much. There is thus hope that the non-core events remain unlikely to happen even though they were not explicitly fixed by the algorithm. Theorem \ref{thm:core} shows that if the LLL-conditions are fulfilled for $\A$ then a non-core event $A \in  \A \setminus  \A'$ is violated in the produced output with probability at most $x_A$. This makes the success probability of such an approach at least $\ds 1 - \sum_{A \in \A \setminus \A'} x_A$.

\begin{theorem} \label{thm:core}
Let $\A' \subseteq \A$ be an efficiently verifiable core subset of $\A$. If there is an $\eps \in [0,1)$ and an assignment of reals $x:\mathcal{A}\rightarrow (0,1)$ such that:
$$\forall A\in \A: \prob{A}\leq (1 - \eps) x(A) \prod_{B \in \Gamma(A) \cap \A'}(1-x(B)).$$
Then the modified MT-algorithm can be efficiently implemented with an expected number of resamplings according to Theorem \ref{thm:runningtime}. The algorithm furthermore outputs a good assignment with probability at least $\ds 1 - \sum_{A \in \A \setminus \A'} x_A$.
\end{theorem}
\begin{proof}
Note that the set $\A'$ on which the actual MT-algorithm is run fulfills the LLL-conditions. This makes Theorem \ref{thm:runningtime} applicable. To argue about the success probability of the modified algorithm, note that $x(A) \geq \prob{A} \prod_{B \in \Gamma'(A)}(1-x(B))$ where $\Gamma'(A)$ are the neighbors of $A$ in the variable sharing graph defined on $A'$. Using Theorem \ref{thm:distrib} we get that the probability that a non-core bad event $A \in \A \setminus \A'$ holds in the assignment produced by the modified algorithm is at most $x_A$. Since core-events are avoided completely by the MT-algorithm a simple union bound over all conditional non-core event probabilities results in a failure probability of at most $\sum_{A \in \A \setminus \A'} x_A$.

Here is furthermore a direct proof of the theorem incorporating the argument from Theorem \ref{thm:distrib} into the proof:

Redefine the witness trees of \cite{moser-tardos:lll} to have only events
from $\A'$ in non-root nodes, thus getting a modification of the
Galton-Watson process from Section 3 of \cite{moser-tardos:lll}.
As in \cite{moser-tardos:lll}, we grow witness trees from an execution-log
starting with a root event that holds at a certain point in time.
This guarantees that we capture events $A \in \A \setminus \A'$ happening
even though they are never resampled (since we never check whether
such events $A$ hold or not). Note that if some $A \in A \setminus \A'$
holds after termination, then there is a witness tree with $A$ as root
and with all non-root nodes belonging to $\A'$. Following the proof of
Lemma 3.1 from \cite{moser-tardos:lll} the probability for this to happen
is at most $\sum_{A \in \A \setminus \A'} x_A$.
(We do not get $\sum_{A \in \A \setminus \A'} x_A/(1 - x_A)$, since $A$
cannot be a child of itself in the witness trees that we construct.)
\end{proof}

While the concept of an efficiently verifiable core is easy to understand, it is not clear how often and how such a core can be found. Furthermore having such a core is only useful if the probability of the non-core events is small enough to make the failure probability, which is based on the union bound over those probabilities, meaningful. The following main theorem shows that in all applications that can tolerate a small ``exponential'' $\eps$-slack as introduced by $\cite{CGH:detlll}$, finding such a good core is straightforward:

\begin{theorem}
\label{thm:polycoremain}
Suppose there is a fixed constant $\eps \in (0,1)$ and an assignment of reals $x:\mathcal{A}\rightarrow (0,1-\eps)$ such that:
$$\forall A \in \A: \prob{A}^{1-\eps} \leq x(A) \prod_{B \in \Gamma(A)}(1-x(B)).$$
Suppose further that $\log 1/\delta \leq \poly(n)$, where $\delta=\min_{A \in \A} x(A) \prod_{B \in \Gamma(A)}(1-x(B))$.
Then for every $p \geq \frac{1}{\poly(n)}$ the set $\{A_i \in \A:~\prob{A_i} \geq p\}$ has size at most $\poly(n)$, and is thus essentially always an efficiently verifiable core subset of $\A$. If this is the case, then there is a Monte Carlo algorithm that terminates after $O(\frac{n}{\eps^2} \log{ \frac{n}{\eps^2}})$ resamplings and returns a good assignment with probability at least $1 - n^{-c}$, where $c > 0$ is any desired constant.
\end{theorem}
\begin{proof}
For a probability $p = 1 /\poly(n)$ to be fixed later we define $\A'$ as the set of events with probability at least $p$. Recall from Theorem \ref{thm:runningtime} that $\sum_{A \in \A} x_A \leq O(n \log(1/\delta))$. Since $x_A \geq p$ for $A \in \A'$, we get that $|\A'| \leq O(n \log(1/\delta) / p) = \poly(n)$. By assumption $\A'$ is efficiently verifiable and we can run the modified resampling algorithm with it.

For every event we have $\prob{A} \leq x_A < 1-\eps$ and thus get a $(1-\eps)^\eps = (1 - \Theta(\epsilon^2))$-slack; therefore Theorem \ref{thm:runningtime} applies and guarantees that the algorithm terminates with high probability after $O(\frac{n}{\eps^2} \log{\frac{n}{\eps^2}})$ resamplings.

To prove the failure probability note that for every non-core event $A \in \A \setminus \A'$, the LLL-conditions with the ``exponential $\eps$-slack'' provide an extra multiplicative $p^{-\eps}$ factor over the LLL-conditions in Theorem \ref{thm:runningtime}. We have $x(A)\prob{A}^{\epsilon} \geq \prob{A} \prod_{B \in \Gamma'(A)}(1-x(B))$ where $\Gamma'(A)$ are the neighbors of $A$ in the variable sharing graph defined on $A'$. Using Theorem \ref{thm:distrib} and setting
$p = n^{-\Theta(1/\eps)}$, we get that the probability that a non-core bad event $A \in \A \setminus \A'$ holds in the assignment produced by the modified algorithm is at most $x_A\prob{A}^{\epsilon}\leq x_{A}n^{-\Theta(1)}$. Since core-events are avoided completely by the MT-algorithm, a simple union bound over all conditional non-core event probabilities results in a failure probability of at most $\frac{1}{n^{\Theta(1)}}\sum_{A \in \A \setminus \A'} x_A$. Now since,
$\sum_{A \in \A \setminus \A'}x_{A} \leq \sum_{A \in A'} x_{A}=T=poly(n)$ holds, we get that we fail with probability at most $n^{-c}$ on non-core events while safely avoiding the core. This completes the proof of the theorem.
\end{proof}

The last theorem nicely completes this section; it shows that in practically all applications of the general LLL it is possible to obtain a fast Monte Carlo algorithm with arbitrarily high success probability. The conditions of Theorem \ref{thm:polycoremain} are very easy to check and are usually directly fulfilled.
That is, in all LLL-based proofs (with a large number of events $A_i$) known to us, the set of high-probability events forms a polynomial-sized core that is trivially efficiently verifiable, e.g., by exhaustive enumeration. Theorem \ref{thm:polycoremain} makes these proofs constructive without further complicated analysis. In most cases, only some adjustments in the bounds are needed to respect the $\eps$-slack in the LLL-condition.

\paragraph{Remarks}
\begin{itemize}
\item  Note that the failure probability can be made an arbitrarily small inverse polynomial. This is important since for problems with non-efficiently verifiable solutions the success probability of Monte Carlo algorithms cannot be boosted using standard probability amplification techniques.
\item In all applications known to us, the core above has further nice structure: usually the probability of an event $A_i$ is exponentially small in the number of variables it depends on. Thus, each event in the core only depends on
$O(\log n)$ many $A_i$, and
hence is usually trivial to enumerate. This makes the core efficiently verifiable, even when finding a general violated event in $\A$ is NP-hard.
\item The fact that the core consists of polynomially many events with usually
logarithmically many variables each, makes it often even possible to enumerate the core in
parallel and to evaluate each event in parallel. If this is the case one can get an RNC algorithm by first building the dependency graph on the core and then computing an MIS of violated events in each round (using MIS algorithms such as \cite{abi,luby:mis}). Using the proof of Theorem \ref{thm:runningtime} which is based on some ideas from the parallel LLL algorithm of MT, it is easy to see
that only logarithmically many rounds of resampling these events are needed.
\item Even though the derandomization of $\cite{CGH:detlll}$ also only requires an ``exponential $\eps$-slack'' in the LLL-conditions, applying the techniques used there and in general getting efficient deterministic algorithms when $m$ is superpolynomial seems hard. The derandomization in  \cite{CGH:detlll} either explicitly works on all $m$ events when applying
the method of conditional probabilities or uses approximate $O(\log{m})$-wise
independent probability spaces which have an inherently $poly(m)$ size domain.
% \item The weighted version of Molloy and Reed applies the LLL with an exponential slack and therefore all proofs building on it and its special case which Molloy and Reed call the asymmetric version can are almost immediately constructive by the above theorem. *Check and proof this :-)*
\end{itemize}

\section{A Constant-Factor Approximation Algorithm for the Santa Claus Problem}
\label{sec:santa}
The {\em Santa Claus} problem  is the restricted assignment version of the max-min allocation problem of indivisible items. In this section, we present the first efficient randomized constant-factor approximation algorithm for this problem.

In the max-min allocation problem,
there is a set $\mathcal{C}$ of $n$ items,
and $m$ children. The value (utility) of item $j$ to child $i$ is
$p_{i,j} \geq 0$. An item can be assigned to only one child. If a child $i$
receives a subset of the items $S_i \subseteq \mathcal{C}$, then the
total valuation of the items received by $i$ is $\sum_{j \in S_i} p(i,j)$.
The goal is to maximize the minimum total valuation of the items received by
any child, that is, to maximize $\min_{i} \sum_{j \in S_{i}} p(i,j)$.
(The ``minmax'' version of this ``maxmin'' problem is the
classical problem of makespan minimization in unrelated parallel
machine scheduling \cite{LST}.) This problem has received much attention
recently \cite{bansal:stoc06,asadpour:stoc07,feige:soda08,asadpour-feige-saberi,bateni:stoc09,julia:focs09,saha-srin:ics10}.

A restricted version of max-min allocation is where each item has an intrinsic value, and where for every child $i$,
$p_{i,j}$ is either $p_{j}$ or $0$. This is known as the Santa Claus problem. The Santa Claus problem is NP-hard
and no efficient approximation algorithm better than $1/2$ can be obtained unless $P=NP$ \cite{bezakova:sigecom05}.
Bansal and Sviridenko \cite{bansal:stoc06} considered a linear-programming (LP) relaxation of the problem
known as the configuration LP, and showed how to round this LP to obtain an $O(\log{\log{\log{m}}}/ \log{\log{m}})$-approximation
algorithm for the Santa Claus problem. They also showed a reduction to a crisp combinatorial problem, a feasible solution to which implies a
constant-factor integrality gap for the configuration LP.

Subsequently, Feige \cite{feige:soda08} showed that the configuration LP has a constant integrality gap. Normally such a proof immediately gives a constant-factor approximation algorithm that rounds an LP solution along the line of the integrality-gap proof. In this case Feige's proof could not be made constructive because
it was heavily based on repeated reductions that apply the asymmetric version of the LLL to exponentially many events. Due to this unsatisfactory situation, the Santa Claus problem was the first on a list of problems reported in the survey
``Estimation Algorithms versus Approximation Algorithms'' \cite{Feige08-estimationvsapproximation} for which a constructive proof would be desirable. Using a completely different approach, Asadpour, Feige and Saberi \cite{asadpour-feige-saberi} could show that the configuration LP has an integrality gap of at most $\frac{1}{5}$. Their proof uses local-search and hypergraph matching theorems of Haxell \cite{haxell1995hypergraphmatching}. Haxell's theorems are again highly non-constructive and the stated local-search problem is not known to be efficiently solvable. Thus this second non-constructive proof still left the question of a constant-factor approximation algorithm open.

In this section we show how our Theorem \ref{thm:polycoremain} can be used to easily and directly constructivize the LLL-based proof of
Feige \cite{feige:soda08}, giving the first constant-factor approximation algorithm for the Santa Claus problem.

It is to be noted that the more general max-min fair allocation problem appears significantly harder. It is known that
for general max-min fair allocation, the configuration LP has a gap of $\Omega(\sqrt{m})$. Asadpour
and Saberi \cite{asadpour:stoc07} gave an $O(\sqrt{m}\ln^{3}(m))$ approximation factor for this problem
using the configuration LP. Recently, Saha and Srinivasan \cite{saha-srin:ics10} have improved this
to $O(\frac{\sqrt{m\ln{m}}}{\ln{\ln{m}}})$. So far the best approximation ratio known for this problem
due to Chakraborty, Chuzhoy and Khanna  is $O(n^{\epsilon})$ \cite{julia:focs09}, for any constant
$\epsilon >0$; their algorithm  runs in $O(n^{1/\epsilon})$ time.

\subsection{Algorithm}

We focus on the Santa Claus problem here. We start by describing the configuration LP and the reduction of it to a combinatorial problem over a set system, albeit with a constant factor loss in approximation. Next we give a constructive solution for the set system problem, thus providing a constant-factor approximation algorithm for the Santa Claus problem.

We guess the optimal solution value $T$ using binary search. An item $j$
is said to be small, if $p_{j} < \alpha T$, otherwise it is
said to be big. Here $\alpha < 1$ is the approximation ratio, which will get fixed later. A configuration is a subset of items.
The value of a configuration $C$ to child $i$ is denoted by $p_{i,C}=\sum_{j \in C} p_{i,j}$. A configuration $C$
is called valid for child $i$ if:
\begin{itemize}
\item  $p_{i,C} \geq T$ and all the items are small; or
\item $C$ contains only one item $j$ and $p_{i,j}=p_{j} \geq \alpha T$, that is, $j$ is a big item for child $i$.
\end{itemize}
Let $C(i,T)$ denote the set of all valid configurations corresponding to child $i$ with respect to $T$.
We define an indicator variable $y_{i,C}$ for each child $i$ and all valid configurations $C \in C(i,T)$ such that
it is $1$ if child $i$ receives configuration $C$ and $0$ otherwise. These variables are relaxed to take
any fractional value in $[0,1]$ to obtain the configuration LP relaxation.

\begin{eqnarray}
\label{eqn:lp}
\forall j: \sum_{C \ni j}\sum_{i} y_{i,C} \leq 1 \\
\forall i: \sum_{C \in C(i,T)} y_{i,C}=1 \nonumber\\
\forall i, C: y_{i,C} \geq 0 \nonumber
\end{eqnarray}

Bansal and Sviridenko showed that if the above LP is feasible, then it is possible to find a fractional allocation
 that assigns a configuration with value at least $(1-\epsilon)T$ to each child in polynomial time.

 The algorithm of Bansal and Sviridenko starts by solving the configuration LP (\ref{eqn:lp}). Then by
 various steps of simplification, they reduce the problem to the following instance:

 \quad

 \emph{ There are $p$ groups, each group containing $l$ children. Each child is associated with a collection of $k$ items
 with a total valuation of $\frac{T}{c}$, for some constant $c >0$. Each item appears in at most $\beta l$ sets for
 some $\beta \leq 3$. Such an instance is referred to as $(k,l,\beta)$-system.}

 The goal is to efficiently select one child from each group and assign at least $\lfloor \gamma k \rfloor$ items
 to each of the chosen children, such that each item is assigned only once. If such an assignment
 exists, then the corresponding $(k,l,\beta)$-system is said to be $\gamma$-good $(k,l,\beta)$-system.

 Feige showed that indeed the $(k,l,\beta)$-system that results from the configuration LP is $\gamma$-good,
 where $\gamma=O\left(\frac{1}{max(1,\beta)}\right)$ \cite{feige:soda08}.
 This established a {\em constant} factor integrality gap for
 the configuration LP. However, the proof being
non-constructive, no algorithm was known to efficiently
 find such an assignment. In the remaining of this section,
 we make Feige's argument constructive, thus giving a constant-factor
 approximation algorithm for the Santa Claus problem. But before that, for the sake of completeness,
 we briefly describe the procedure that obtains a $(k,l,\beta)$-system from an optimal solution of the
 configuration LP \cite{bansal:stoc06}.

\subsection{From a configuration LP solution to a $(k,l,\beta)$-system}

The algorithm starts by simplifying the assignment of {\em big} items in an optimal solution (say) $y^{*}$ of the configuration LP.
Let $J_{B}$ denote the set of big items. Consider a bipartite graph $G$ with children $M$ on the right side and big items $J_{B}$ on the left side.
An edge $(i,j), i \in M, j \in J_{B}$ of weight $w_{i,j}=\sum_{j \in C(i,T)} y^{*}_{i,C}$ is inserted in $G$ if  $w_{i,j} > 0$. These
$w_{i,j}$ values are then modified such that after the modification the edges of $G$ with weight in $(0,1)$ form a forest.

{\em {\it \bf Lemma 5 \cite{bansal:stoc06}.}
The solution $y^{*}$ can be transformed into another feasible solution of the configuration LP where the graph $G$ is a forest.
}

The transformation is performed using the simple cycle-breaking trick. Each cycle is broken into two matchings;
weights on the edges of one matching are increased gradually while the weights on the other are decreased until
some weights hit $0$ or $1$. If a $w_{i,j}$ becomes $0$ in this procedure,
the edge $(i,j)$ is removed from $G$. Else if it becomes $1$, then item $j$ is permanently assigned to child $i$ and
the edge $(i,j)$ is removed.

Suppose $G'$ is the forest obtained after this transformation. The forest structure is then further exploited to form
groups of children and big items.

{\em {\it \bf Lemma 6 \cite{bansal:stoc06}.}
The solution $y^{*}$ can be transformed into another solution $y'$ such that children  $M$ and big items $J_{B}$ can be clustered into
$p$ groups $M_1, M_2, \ldots, M_{p}$ and $J_{B,1}, J_{B_2}, \ldots, J_{B_p}$ respectively with the following properties.
\begin{enumerate}
\item For each $i=1,2,\ldots,p$, the number of jobs $J_{B,i}$ in group $M_i$ is exactly $|M_{i}|-1$. The group $J_{B,i}$ could
possibly be empty.
\item Within each group the assignment of big job is entirely flexible in the sense that they can be placed feasibly on any of the
$|M_{i}|-1$ children out of the $|M_{i}|$ children.
\item For each group $M_{i}$, the solution $y'$ assigns exactly one unit of small configurations to children in $M_{i}$ and all the
$|M_{i}|-1$  units of configurations correspond to big jobs in $J_{B,i}$. Also, for each small job $j$, $\sum_{C \ni j}\sum_{i} y'_{i,C} \leq 2$.
\end{enumerate}
}

Lemma 6 implies that the assignment of big items to children in a group is completely flexible and can be ignored. We only need to
choose one child from each group who will be satisfied by a configuration of small items. Let $y'$ assigns a small configuration
$C$ to an extent of $y'_{m,C}$ to some child $c \in M_{i}, i\in [1,p]$, then we say that $M_{i}$ contains the small configuration $C$ for
child $c \in M_{i}$. Without loss of generality, it can be assumed that each child in the groups is fractionally assigned to exactly
one small configuration. Bansal and Sviridenko further showed that $y'$ can again be simplified such that each small configuration is assigned
to at least to an extent of $\frac{1}{l}=\frac{1}{n+m}$ to each child and for each small job $j$,  $\sum_{C \ni j}\sum_{i} y'_{i,C} \leq 3$. This
implies, if we consider all the small configurations across $p$ groups, then each small job appears in at most $\beta l$ configurations, where $\beta=3$.

Finally, the following lemma shows that by losing a constant factor in the approximation, one can assume that all the small jobs have same size.

{\em {\it \bf Lemma 8 \cite{bansal:stoc06}.}
Given the algorithmic framework above, by losing a constant factor in the approximation, each small job can be
assumed to have size $\frac{\epsilon T}{n}$.
}

As a consequence of the above lemma, we now have the following scenario.

{\em There are $p$ groups $M_{1}, M_2, \ldots, M_{p}$, each containing at most $l$ children. Each child is associated with a set that contains
$k=\Theta(\frac{n}{\epsilon})$ items. Each item belongs to at most $\beta l$ sets. The goal is to pick one child from each group and assign at least
a constant fraction of the items in its set such that each item is assigned exactly once.}

Therefore, we arrive at what is referred as a $(k,l,\beta)$-system.

\subsection{Construction of a $\gamma$-good solution for a $(k,l,\beta)$-system}

We now point out the main steps in Feige's algorithm, and in detail, describe the modifications required
 to make Feige's algorithm constructive.

 \paragraph*{Feige's Nonconstructive Proof for $\gamma$-good $(k,l,\beta)$-system:}

Feige's approach is based on a systematic reduction of $k$ and $l$ in iterations, finally arriving to a system
 where $k$ or $l$ are constants. For constant $k$ or $l$ the following lemma asserts a constant $\gamma$.

 \begin{lemma}[Lemma 2.1 and 2.2 of \cite{feige:soda08}]
 \label{lemma:constant}
 For every $(k,l,\beta)$-system a $\gamma$-good solution with $\gamma$ satisfying, $\gamma=\frac{1}{k}$ or
 $\gamma k =\lfloor \frac{k}{\lceil \beta l \rceil} \rfloor$ can be found efficiently.
  \end{lemma}

The reduction of $(k, l, \beta)$-system to constant $k$ and $l$ involves two main lemmas, which we refer to
 as {\em Reduce-l} lemma and {\em Reduce-k} lemma respectively.

 \begin{lemma}[Lemma 2.3 of \cite{feige:soda08}, Reduce-l]
 For $l > c$ ($c$ is a sufficiently large constant), every $\gamma$-good $(k,l,\beta)$-system with $k \leq l$ can
 be transformed into a $\gamma$-good $(k,l', \beta')$-system with $l' \leq \log^{5}{l}$ and $\beta' \leq \beta(1+\frac{1}{\log{l}})$.
 \end{lemma}

 \begin{lemma}[Lemma 2.4 of \cite{feige:soda08}, Reduce-k]
 Every $(k,l,\beta)$-system with $k \geq l \geq c$ can be transformed into a $(k',l,\beta)$-system with $k' \leq \frac{k}{2}$
 and with the following additional property: if the original system is not $\gamma$-good, then the new system is not
 $\gamma'$-good for $\gamma'=\gamma(1+\frac{3\log{k}}{\sqrt{\gamma k}})$. Conversely, if the new system is $\gamma'$-good,
 then the original system was $\gamma$-good.
 \end{lemma}

 If  $\beta$ is not a constant to start with, then by applying the following lemma repeatedly, $\beta$
 can be reduced below $1$.

 \begin{lemma}[Lemma 2.5 of \cite{feige:soda08}]
 For $l > c$, every $\gamma$-good $(k,l, \beta)$-system can be transformed into a $\gamma$-good $(k', l, \beta')$-system with
 $k'=\lfloor \frac{k}{2} \rfloor$ and $\beta' \leq \frac{\beta}{2}\left(1+\frac{\log{\beta l}}{\sqrt{\beta l}}\right)$.
 \end{lemma}

 However in our context, $\beta \leq 3$, thus we ignore Lemma 2.5 of \cite{feige:soda08} from further
 discussions.

 Starting from the original system, as long as $l > c$, Lemma Reduce-l is applied when $ l > k$ and  Lemma Reduce-k
 is applied when $k \geq l$. In this process $\beta$ grows at most by a factor of $2$. Thus at the end, $l$ is a constant and so is $\beta$.
 Thus by applying Lemma \ref{lemma:constant}, the constant integrality gap for the configuration LP
 is established.

 \paragraph*{Randomized Algorithm for $\gamma$-good $(k,l,\beta)$-system:}

 There are two main steps in the algorithm.

 \begin{enumerate}
 \item Show a constructive procedure to obtain the reduced system through Lemma Reduce-l and Lemma Reduce-k.
 \item Map the solution of the final reduced system back to the original system.
 \end{enumerate}

 We now elaborate upon each of these.

%Recall from the constructive version of LLL given by  \cite{moser-tardos:lll} (description in Section \ref{section:result}): If the algorithm detects a violated event, it resamples the underlying  random variables associated with the violated event and proceed. On expectation the number of resampling steps required to avoid a bad event $A$ is $\frac{x_{A}}{(1-x_{A})}$; and thus on expectation $\sum_{A} \frac{x_{A}}{(1-x_{A})}$  steps are enough to avoid all the bad events. The algorithm is order independent, that is, it does  not matter in which order the bad events are detected and rectified. From now onward we refer to the algorithm by Moser and Tardos in \cite{moser-tardos:lll} by MT.

 \subsubsection{Making Lemma Reduce-l Constructive}

This follows quite directly from \cite{moser-tardos:lll}. The algorithm
picks $\lfloor \log^{5}{l} \rfloor$ sets uniformly at random and
independently from each group. Thus while the value of $k$ remains fixed,
$l$ is reduced to $l'=\lfloor \log^{5}{l} \rfloor$. Now in expectation the value of
$\beta$ does not change and the probability that
$\beta' > \beta(1+\frac{1}{\log{l}})$, and hence
$\beta'l' > \beta l (1+\frac{1}{\log{l}})$, is at most
$e^{-\beta' l' /3 \log^{2}{l}} \leq e^{-\log^{3}{l}}=l^{-\log^{2}{l}}$.
We define a bad event corresponding to each element:
 \begin{itemize}
 \item $A_{j}$: Element $j$ has more than $\beta' l'$ copies.
 \end{itemize}

Now noting that the dependency graph has degree at most
$kl \beta l \leq 6 l^3$, the uniform (symmetric) version of the LLL applies.
Now it is easy to check if there exists a violated event: we simply
count the number of times an  element appears in all the sets.
Thus we directly follow \cite{moser-tardos:lll};
setting $x_{A_{j}}=\frac{1}{el^{\log^{2}{l}}}$, we
get the expected running time to avoid all the bad events to be
$O(plk/l^{\log^{2}{l}})=O(p)=O(m)$.

\subsubsection{Making Lemma Reduce-k Constructive} \label{subsec:reducek}

This is the main challenging part.
The random experiment involves  selecting each item independently at random with probability $\frac{1}{2}$. To  characterize the bad events, we need a structural lemma from \cite{feige:soda08}. Construct a graph on the sets, where there is an edge between two sets if they share an element. A collection of sets is said to be connected if and only if the subgraph induced by this collection is connected.

We consider two types of bad events:
 \begin{enumerate}
 \item $B_{1}$: some set has less than $k'=\left(1-\frac{\log{k}}{\sqrt{k}}\right)\frac{k}{2}$ items surviving, and
 \item  $B_{i}$ for $i \geq 2$: there is a connected collection of $i$ sets from distinct groups whose union originally contained at most
 $i \gamma k$ items, of which more than $i \delta' \frac{k}{2}$ items survive, where $\delta'=\gamma\left(1+\frac{\log{k}}{\sqrt{\gamma k}}\right)$.
 \end{enumerate}
%(The parameter $\lambda$ is denoted $\delta'$ in \cite{feige:soda08}.)

If none of the above bad events happen, then we can consider the first $k'$ items from each set and yet the second type of bad events do not happen. These events are chosen such that $\gamma'$-goodness ($\gamma'=\delta'\frac{k}{2}\frac{1}{k'}\leq \gamma\left(1+\frac{3\log{k}}{\sqrt{\gamma k}}\right)$) of the new system certifies that the original system was $\gamma$-good. That this is indeed the case follows directly from Hall's theorem as proven by Feige:

\begin{lemma}[Lemma 2.7 of \cite{feige:soda08}]
Consider a collection of $n$ sets and a positive integer $q$.
\begin{enumerate}
 \item If for some $1 \leq i \leq n$, there is a connected subcollection of $i$ sets whose union
 contains less than $iq$ items, then there is no choice of $q$ items per set such that all items
 are distinct.
 \item If for every $i$, $ 1 \leq i \leq n$, the union of every connected subcollection of $i$
 sets contains at least $iq$ (distinct) items, then there is a choice of $q$ items per set such
 that all items are distinct.
\end{enumerate}
\end{lemma}

Feige showed in \cite{feige:soda08} that for bad events of type $B_i, i \geq 1$, taking $x_{i}=2^{-10i\log{k}}$ is sufficient to satisfy the condition (\ref{eqn:alll}) of the asymmetric LLL. More precisely, suppose we define, for any
bad event $B \in \bigcup_{i \geq 1} B_i$, $\Gamma(B)$ to be as in
Section~\ref{sec:algframework}: i.e., $\Gamma(B)$ is the set of all
bad events $A \neq B$ such that $A$ and $B$ both depend on at least one
common random variable in our ``randomly and independently selecting items''
experiment. Then, it is shown in \cite{feige:soda08} that with the choice
$x_{i}=2^{-10i\log{k}}$ for all events in $B_i$, we have
\begin{equation}
\label{eqn:feige-santa}
\forall (i \geq 1) ~\forall (B \in B_i), ~
\prob{B} \leq 2^{-20 i \log{k}} \leq
x_i \prod_{j \geq 1} \prod_{A \in (B_j \cap \Gamma(B))} (1 - x_j).
\end{equation}
Thus by the LLL, there exists an assignment that avoids all the bad events. However, no efficient construction was known here, and as Feige points out, ``the main source of difficulty in this respect is Lemma 2.4, because there the number of bad events is exponential in the problem size, and moreover, there are bad events that involve a constant fraction of the random variables.'' Our Theorem \ref{thm:polycoremain} again directly makes this proof constructive and gives an efficient Monte Carlo algorithm for producing a reduce-k system with high probability.

\begin{lemma} \label{lemma:good}
There is a Monte Carlo algorithm that produces a valid reduce-k system with
probability at least $1 - 1/m^2$.
\end{lemma}
\begin{proof}
Note from (\ref{eqn:feige-santa}) that we can take $\delta = 2^{-20 m \log{k}}$. So, we get that $\log 1/\delta = O(m \log k) = O(n \log n)$ where $n$ is the number of items and $m \leq n$ is the number of children. We furthermore get that all events with probability larger than a fixed inverse-polynomial involve
only connected subsets of size $O(\frac{\log m}{\log k})$ and Theorem \ref{thm:polycoremain} implies that there are only polynomially many such ``high'' probability events. (This can also be seen directly since the degree of a subset
is bounded by $k\beta l \leq 6k^2$ and the number of connected subcollections is therefore at most $(6k^2)^{O(\frac{\log m}{\log k})} = m^{O(1)}= n^{O(1)}$.) The connected collections of subsets are easy to enumerate
using, e.g., breadth-first search and are therefore efficiently verifiable (in fact, even in parallel). Theorem \ref{thm:polycoremain} thus applies and directly proves the lemma.
\end{proof}

\subsubsection{Mapping the solution of the final reduced system back to the original system}

By repeatedly applying algorithms to produce Reduce-l or Reduce-k system, we can
completely reduce down the original system to a system with a constant number of children per group, where
$\beta$ can increase from $3$ to at most $6$ due to Lemma Reduce-l.
This involves at most $\log m$ Reduce-l reductions and at most $\log n$ Reduce-k reductions.
We can furthermore assume that $n < 2^m$ since otherwise simply all combinations of one child per group could be tried in time polynomial in $n$.
Since, each Reduce-l or Reduce-k operation produces a desired solution with probability at least $1-\frac{1}{m^{2}}$, by union bound,
with probability at least $1 - O(\log n \log m/m^2) = 1 - O(\log m / m)$ a final $(k,l,\beta)$-system is produced that is $\gamma$-good for some
constant $\gamma$ by Lemma \ref{lemma:constant}. Using Lemma \ref{lemma:constant}, we can also find a $\gamma$-good selection of children. Now, once
one child from each group is selected, we can construct a standard network flow instance to assign items to these chosen children (Lemma \ref{lemma:assign}).
This finishes the process of mapping back a solution of the reduced system to the original $(k,l,\beta)$-system. While checking whether
an individual reduction failed seems to be a NP-hard task, it is easy to see in the end whether a good enough assignment is produced. This enables us to rerun the algorithm in the unlikely event of a failure. Thus, the Monte Carlo algorithm can be strengthened to an algorithm that always produces a good solution and has an expected polynomial running-time.

The details of the above are given in two lemmas, Lemma \ref{lemma:constant-new} and Lemma \ref{lemma:assign}.
Theorem \ref{theorem:santa} follows from the two lemmas.

Suppose we start with a $(k_1, l_1, \beta_1)$-system and after repeated application of either
Lemma Reduce-l or Lemma Reduce-k reach at a $(k_{s}, l_{s}, \beta_{s})$-system, where
$l_{s} < c$, a constant. We then employ Lemma \ref{lemma:constant} to obtain a $\gamma_{s}$-good
$(k_s, l_s, \beta_{s})$-system, where $\gamma_{s}$ satisfies
$\gamma_{s} k_{s} =\lfloor \frac{k_{s}}{\lceil \beta_{s} l_{s} \rceil} \rfloor$. Since $l_{s}$ is a constant
and $\beta_{s} \leq 6$, $\gamma_{s}$ is also a constant. Lemma \ref{lemma:constant} also gives a choice of a child
from each group, denoted by a function $f : \{1,\ldots,p\} \rightarrow \{1, \ldots, l_{s}\}$ that serves as
a witness for $\gamma_{s}$-goodness of $(k_{s}, l_{s}, \beta_{s})$-system. We use this same mapping for the original system.
The following lemma establishes the goodness of the $(k_1,l_1,\beta_{1})$-system.

\begin{lemma}
\label{lemma:constant-new}
Given a sequence of reductions of $k$, $(k_1,l_1,\beta_1)\rightarrow \ldots \rightarrow (k_{s}, l_{s}, \beta_{s})$,
interleaved with reductions of $l$, let for all $s \geq 2, \gamma_{s}=\gamma_{s-1}(1+\frac{3\log{k_{s-1}}}{\sqrt{\gamma_{s-1} k_{s-1}}})$.
Then if the final reduced system is $\gamma_{s}$-good and the function $f: \{1,\ldots,p\} \rightarrow \{1,\ldots, l_{s}\}$
serves as a witness for its $\gamma_{s}$-goodness, then $f$ also serves as a witness of $\gamma$-goodness of $(k_1,l_1,\beta_1)$ system
with high probability. In other words, we can simply use the assignment given by $f$ to select one child from each group and that assignment serves
 as a witness of $\gamma$-goodness of the original system with high probability.
\end{lemma}

\begin{proof}
Suppose there exists a function $f$ that serves as a witness for
$\gamma_{s}$-goodness of the $(k_s,l_s,\beta_{s})$-system, but does not
serve as a witness that $(k_{s-1},l_{s-1},\beta_{s-1})$-system is $\gamma_{s-1}$-good. Then there must exist  a connected
collection of $i, i >0$ sets chosen from $p$ groups according to $f$, such that their union contains less than $\gamma_{s-1}k_{s-1}i$
items. However in the reduced system, their union has $\gamma_{s}k_{s-1}i$ elements. Call such a function $f$ bad. Thus every
bad function is characterized by a violation of event of type
$B_{i}, i \geq 1$, described in Section~\ref{subsec:reducek}.
However, by Lemma \ref{lemma:good} we have
$\prob{\exists \text{ a bad function } f }\leq
\prob{ \text{an event of type } B_{i}, i \geq 1 \text{ happens}}
\leq \frac{1}{m^2}$.

Now the maximum number of times the Reduce-k step is applied is at most
$\log{k_1}\leq \log{n}$. Thus if the Reduce-l step is not applied at all,
then by a union bound, function $f$ is $\gamma$-good for the
$(k_1,l_1,\beta_1)$-system with probability at least
$1-\frac{\log{m}\log{n}}{m^2}$. We can assume
 without loss of generality that $n \leq 2^m$. (Otherwise in polynomial time we can guess the children who receive small items and thus know $f$.
 Once $f$ is known, an assignment of small items to the children chosen by $f$ can be done in polynomial time through Lemma \ref{lemma:assign}.)
 Since $n \leq 2^m$, function $f$ is $\gamma$-good for the
$(k_1,l_1,\beta_1)$-system with probability at least $1-\log{m}/m$. Now since
the Reduce-l step only reduces $l$ and keeps $k$ intact, it does not affect the goodness of the set system.
\end{proof}

Once we know the function $f$, using Lemma \ref{lemma:assign}, we can get
a valid assignment of $\floor{k\gamma}$ items to each chosen child:
\begin{lemma}
\label{lemma:assign}
Given a function $f:\{1,\ldots, p\} \rightarrow \{1,\ldots,l\}$, and parameter $\gamma$, there is a polynomial time algorithm to determine,
whether $f$ is $\gamma$-good and we can determine the subset of $\floor{k\gamma}$ items received by each child $f(i), i\in[1,p]$.
\end{lemma}

\begin{proof}
We construct a bipartite graph with a set of vertices $U=\{1,\ldots,p\}$
corresponding to each chosen child from the $p$ groups, a set of vertices
$V$ corresponding to the small items in the sets of the chosen children,
a source $s$ and a sink $t$.
Next we add a directed edge of capacity $\floor{\gamma k}$ from source $s$ to each vertex in $U$. We also add directed edges $(u,v), u \in U, v \in V$, if the item $u$ belongs to the set of $v$. These edges have capacity $1$. Finally we add a directed edge from each vertex in $V$ to the sink $t$ with capacity $1$. We claim that this flow network has a maximum flow of $\floor{k\gamma}p$ iff $f$ is $\gamma$-good:

For the one direction let $f$ be $\gamma$-good. Thus there exists a set of $\floor{\gamma k}$ elements that can be assigned to each child $u \in U$. Send one unit of flow from each child to these items that it receives. The outgoing flow from each $u \in U$ is exactly $\floor{\gamma k}$. Since each item is assigned to at most one child, flow on each edge $(v,t), v \in V$ is at most $1$. Thus all the capacity constraints are maintained and the flow value is $\floor{\gamma k}p$.

For the other direction consider an integral maximum flow of $\floor{k\gamma}p$. Since the total capacity of all the edges emanating from the source is $\floor{k\gamma}p$, they must all be saturated by the maxflow. Since the flow is integral, for each child $u$ there are exactly $\floor{\gamma k}$ edges with flow $1$ corresponding to the items that it receives. Also since no edge capacity is violated, each item is assigned to exactly one child. Therefore $f$ is $\gamma$-good.

To check a function $f$ for $\gamma$-goodness and obtain the good assignment we construct the flow graph and run a max flow algorithm that outputs in an integral flow. As proven above a max flow value of $\floor{k\gamma}p$ indicates $\gamma$-goodness and for a $\gamma$-good $f$ the assignment can be directly constructed from the flow by considering only the flow carrying edges.
\end{proof}

%Hence finally get a Monte Carlo algorithm that w.h.p. produces an assignment which is within a constant of the optimum. While checking whether an individual reduction failed seems to be a (NP-)hard task it is easy to check whether in the end a good assignment is produced through the rounding procedure. This enables us to rerun the algorithm in the unlikely event of a failure and thus the Monte Carlo algorithm can be strengthened to an algorithm that always produces a good solution and has an expected polynomial running-time:

\begin{theorem}
\label{theorem:santa}
There exists a constant $\alpha > 0$ and a randomized algorithm for the Santa Claus problem that runs in expected polynomial time and always assigns items of total valuation at least $\alpha \cdot \mathrm{OPT}$ to each child.
\end{theorem}

\section{Non-repetitive Coloring of Graphs}
\label{sec:non-rep}

In this section, we give an efficient Monte-Carlo
construction for non-repetitive coloring of graphs.
Call a word (string) $w$ ``squarefree'' or ``non-repetitive'' if there does
not exist
any string of the form $w=xx$, where $x \neq \emptyset$.
Let us refer to graphs using the symbol $H$ instead of $G$, to not confuse
with our dependency graphs $G$.
Recall from Section \ref{sec:intro} that a $k$-coloring of the edges of $H$
is called \emph{non-repetitive} if the sequence of colors
along any path in $H$ is squarefree: i.e., we want a coloring in which
no path has a color-sequence of the form $xx$. (All paths here refer to
simple paths.) The smallest $k$ such that
$H$ has a non-repetitive coloring using $k$ colors is called
the {\em Thue number} of $H$ and
is denoted by $\pi(H)$.
The Thue number was first defined by Alon, Grytczuk, Hauszczak and Riordan
in \cite{alon:random02}: it is named after Thue who proved in 1906 that
if $H$ is a simple path, then $\pi(H)=3$ \cite{thue:1906}.
While the method of Thue is constructive, no efficient construction
is known for general graphs. Alon et al.\ showed through
application of the asymmetric LLL that $\pi(H)\leq c \Delta(H)^{2}$ for
some absolute constant $c$. Their proof was non-constructive.
The number of bad events is exponential. Not only that, checking whether a given
coloring is non-repetitive is coNP-Hard, even when the number of colorings is restricted to $4$
 \cite{marx:discrete09}. Thus checking if some ``bad event'' holds in
a given coloring is coNP-Hard.
Since the work of Alon et al.,
the non-repetitive coloring of graphs has received a good deal of attention
in the last few years
\cite{james:thr05, marcus:sigact02, jaroslaw:discrete, kundgen:08, bresar:07, alonG:08}.
Yet no efficient construction is known till date, except for some
special classes of graphs such as complete graphs, cycles and trees.

\subsection{Randomized algorithm for obtaining a non-repetitive coloring}
Suppose we are given a graph $H$ with maximum degree $\Delta$.
We first give the proof of Alon et al.\ which shows that
$\pi(H) \leq c\Delta^2$, and then show how to convert this proof directly
into a constructive algorithm (with the loss of a $\Delta^{\epsilon}$ factor
in the number of colors used):

\begin{theorem}[Theorem 1 of \cite{alon:random02}] \label{thm:nonrepcol}
There exists an absolute constant $c$ such that
$\pi(H) \leq c \Delta^{2}$
for all graphs $H$ with maximum degree at most $\Delta$.
\end{theorem}

\begin{proof}
Let $C = (2e^{16}+1) \Delta^2$. Randomly color each edge of $H$ with colors from $C$.
Consider the following types of bad events $B_{i}$, for $i \geq 1$:
``there exists a path $P$ of length $2i$, such that the second
half of $P$ is colored identically to its first half''.

We have for a path $P$ of length $2i, i \geq 1$,
$\prob{\text{P has coloring of the form xx} }= \frac{1}{C^{i}}$.
Also, a path of length $2i$ intersects at most $4ij\Delta^{2j}$ paths of
length $2j$. Thus, for any bad event $A$ of type $i$, we have
$\prob{A} = \frac{1}{C^{i}}$ and that each bad event of type $i$
share variables with at most $4ij\Delta^{2j}$ bad events of
type $B_j$. Set $x_{i}=\frac{1}{2^i \Delta^{2i}}$.
We have $(1-x_j)\geq e^{-2x_{j}}$; this, along with the fact that
$\sum_{j \geq 1} j/2^j = 2$, shows that
$$x_{i} \prod_{j} (1-x_{j})^{4ij\Delta^{2j}} \geq x_{i} e^{-8i  \sum_{j} x_{j}j\Delta^{2j}}\geq \frac{1}{2^i \Delta^{2i}}e^{-8i\sum_{j} \frac{j}{2^{j}}} = (2 e^{16} \Delta^2)^{-i}.$$
Since $C=(2e^{16}+1) \Delta^2$, the condition of the LLL is satisfied and
we are guaranteed the existence of such a non-repetitive coloring.
\end{proof}

Now we see that using just a slightly higher number of colors suffices to make Theorem \ref{thm:polycoremain} apply.

\begin{theorem}
There exists an absolute constant $c$ such that for every constant
$\eps > 0$ there exists a Monte Carlo algorithm that given a graph $H$
with maximum degree $\Delta$, produces a non-repetitive coloring using
at most $c \Delta^{2+\eps}$ colors. The failure probability of the algorithm
is an arbitrarily small inverse polynomial in the size of $H$.
\end{theorem}
\begin{proof}
We apply the LLL using the same random experiments and bad events as in
Theorem~\ref{thm:nonrepcol} but with $C' = C^{\frac{1}{1-\eps'}}$ colors
such that $C' < c \Delta^{2+\eps}$. Using the same settings for $x_A$
gives an exponential $\eps'$ slack in the LLL-conditions since the
probability of a bad event of type $i$ is now at most
$\frac{1}{C'^{i}} = \left(\frac{1}{C'^{i}}\right)^{\frac{1}{1-\eps'}}$.
Recall Theorem \ref{thm:polycoremain}.
Clearly, $\log 1/\delta = O(n^2)$ and so the last thing to check to apply
Theorem \ref{thm:polycoremain} is that for any inverse polynomial $p$, the bad events
with probability at least $p$ are efficiently verifiable. Here these events
consist of paths smaller than a certain length (of the form
$O((1/\epsilon) \log n / \log \Delta)$, where $n$ is the number
of vertices), and Theorem \ref{thm:polycoremain}
guarantees that there are only polynomially many of these. Using
breadth-first-search to go through these paths and checking each of
them for non-repetitiveness is efficient and thus
Theorem \ref{thm:polycoremain} directly applies.
\end{proof} 
\section{Ramsey-type Bounds}
\label{sec:ramsey}
In this section, we briefly sketch another application of our method,
namely the construction of Ramsey-type graphs.

The Ramsey number $R(K^s,K^t)$ is the smallest number $\ell$ such that
for any $n \geq \ell$ and in any
red-blue coloring of the edges of $K^{n}$, there either exists a $K^s$ with all red edges or a $K^t$ with all blue edges.
Here, $K^{a}$ for any integer $a$ denotes a clique of size $a$ as usual.
The fact that these numbers are finite for all $s,m$ is a special case of
Ramsey's well-known theorem (see e.g. \cite{ramsey:book}).  In one of the first
applications of probabilistic methods in combinatorics Erd$\ddot{\text{o}}$s showed the lower bound
of $R(K^m,K^m)=\Omega(m2^{m/2})$
\cite{Erdos47}. Since, then obtaining lower bounds on $R(K^s,K^t)$ and constructing Ramsey graphs
avoiding ``large'' cliques as well as ``large'' independent sets
simultaneously has attracted much attention
\cite{alon94,alon99,kriv2,alonkriv}. The case of fixed $s$ is the main
example case for {\em off-diagonal} Ramsey numbers. Alon and Pudl\'{a}k gave
an explicit deterministic construction for off-diagonal Ramsey graphs in \cite{alon94}.
They showed constructively for some $\epsilon > 0$,
$R(K^s,K^t) \geq t^{\epsilon\sqrt{\log{s}/\log{\log{s}}}}$. The
best known bound for $R(K^s,K^t)$ can be obtained using
LLL, $R(K^s,K^t) =O\left(\frac{t}{\log{t}}\right)^{\frac{s+1}{2}}$ \cite{alonkriv}.
Krivelevich  gave a Monte Carlo algorithm matching this bound
through large deviation inequalities \cite{kriv2}. In addition, Krivelevich
considered related Ramsey type problems, for example, he showed that there
exists a $K^{4}$-free graph on $n$ vertices in which any $o(n^{3/5}\log^{1/2}{n})$ set of vertices
does not contain a $K^{3}$. The problem of finding constructions for Ramsey type graphs matching the best known bounds is
of great interest and may have algorithmic applications as well.

Using our method, we can achieve the best known bound for off-diagonal Ramsey numbers, that is, for fixed $s$, by directly
making the LLL-based proof \cite{alonkriv} constructive. More importantly,
we can provide randomized (Monte Carlo) constructions of graphs on $n$
vertices of the form: ``there is no subgraph $U$ in any set of $s$ vertices
and no subgraph $W$ in any set of $t$ vertices'', where $t$ can
be large, typically $n^{\Theta(1)}$ -- the existence of which can be
proved using the LLL (often using appropriate random-graphs $G(n,p)$).
When $U=K^{s}$, $W=K^{t}$ and $s$ is fixed, we get the off-diagonal
Ramsey number. We refer to these as general Ramsey-type graphs. This is a direct generalizations of
related Ramsey type problems considered by Krivelevich \cite{kriv2}.

When $U$ and $W$ are some special subgraphs, few results are known.
As mentioned earlier Krivelevich considered the case where $U=K^{4}$, $W=K^{3}$,
$s=4$, and $t=o(n^{3/5}\log^{1/2}{n})$. In addition, he also showed constructions for arbitrary
$U$, but when $W=K^{t}$ \cite{kriv2}. Alon  and  Krivelevich in
\cite{alonkriv} and Krivelevich in \cite{kriv2} considered a Ramsey-type bound
$R'(K^s,K^{r}(t))$: the smallest number $\ell$ such that for any
$n \geq \ell$, any graph
on $n$ vertices either contains a $K^{s}$ or there exists a set of
$t$ vertices containing a $K^{r}$. When $r=2$, $R(K^s,K^t)=R'(K^s,K^{2}(t))$. However, to the best of
our knowledge, no general algorithmic result avoiding any subgraph $U$ and $W$ on
any set of
$s$ and $t$ vertices respectively is known till date.

Briefly, the idea is as follows. Suppose, as in the typical
existence-proofs for such graphs, we are able to show using the
(asymmetric) LLL that for a suitable $p = p(n)$, the random graph with $n$
vertices and independent edge-probability $p$, satisfies all the required
properties with positive probability. Theorem~\ref{thm:polycoremain} will typically
immediately apply if we allow an exponential $\eps$-slack. When $s$ is fixed,
another related approach is to apply Theorem~\ref{thm:core} and to only
verify the events that correspond to $s$-sized subgraphs; since $s$ is
fixed, these can be enumerated and verified efficiently. Note that as pointed out
in \cite{alon-spencer}, the LLL may be much more significant in improving the
bounds with fixed $s$ and this is generally the case of interest while applying the LLL based
arguments.

\section{Acyclic Edge Coloring}
\label{sec:acyclic}
Given a graph $G$, an acyclic edge coloring is a proper edge coloring of $G$ where each cycle receives more than $2$ colors. In a proper edge coloring, no two incident edges receive the same color. In addition, here we require that no cycle receives only $2$ colors. The goal is to use a minimum number of colors  (known as the acyclic chromatic number $a(G)$) and obtain an acyclic
edge coloring. The concept of $a(G)$ was introduced way back in 1973 \cite{baum:acyclic} and has been studied by a series of researchers over the years \cite{alon:acyclic,molloy:lll,alon:acyclic1,baum:acyclic,muthua:acyclic}. In all these works, the asymmetric LLL is applied to achieve the best non-constructive bounds. Thus an algorithmic version of the local lemma strikes as the first choice to obtain an acyclic edge coloring.

Alon, McDiarmid and Reed \cite{alon:acyclic} showed that $a(G) < 64\Delta$, where $\Delta$ is the maximum degree. The constant was later improved to $16$ by Molloy and Reed \cite{molloy:lll}, but the proof still was non-constructive.
Both the methods are essentially the same: \emph{ randomly color each edge from a pool of colors $\{1,2,\ldots,C\}$}. They define a series of bad events, where \emph{Type 1} bad event corresponds to two incident edges $e,f$ receiving the same color and \emph{Type $k$} bad event implies a cycle of length $2k$ getting $2$ colors. A cycle of odd length automatically gets $3$ colors, if the coloring is proper. Note that the number of Type $k$ events for non-constant $k$ is super-polynomial in the number of edges of $G$. The probability of Type $1$ event is $1/C$ and the probability of Type $k$ event is $1/C^{2(k-1)}$. Let $\Delta$ be the maximum degree of $G$. It is now an easy exercise to verify that each Type $k$ event depends on at most $4k\Delta$ Type $1$ events and $2k\Delta^{2(l-1)}$ Type $l$ events. With this dependency, setting $C=16\Delta$, $x_{e}=\frac{2}{C}$ for each edge $e$ and $x_{k}=\left(\frac{2}{C}\right)^{2(k-1)}$ for each cycle of length $2k$ satisfies the asymmetric LLL condition \ref{theorem:lll}. We can turn this proof to an algorithm using $16\Delta$ colors as a direct corollary of Theorem \ref{thm:runningtime}.

\begin{theorem} \label{thm:acyclic1}
There is a randomized algorithm that produces a valid acyclic coloring of any graph with $n$ edges and maximum degree $\Delta$ in expected polynomial time using $16\Delta$ colors.
\end{theorem}
\begin{proof}
Recall from Section \ref{sec:algframework} that $\delta$ needs to be at least as high as the smallest upper bound on a probability of a bad event which is $1/C^{2(k-1)}$ for an event of Type $k$. This gives that $\log 1/\delta$ is at most $O(n \log \Delta)$. Further, it is easy to check for a violated bad event in $O(\Delta^2 n)$ time: consider subgraphs on every pair of colors and check if there is a cycle in it. Therefore Theorem \ref{thm:runningtime} directly applies (we can set $\epsilon=0$ in it) and we get that the expected running time is $O(\Delta^2 n) \cdot n \cdot O(n \log{\Delta}) = O(n^3 \Delta^2 \log{\Delta})$. Note that this is far from tight. We can, for example, exploit that there is already a $\eps$-slack in the analysis to get a smaller number of resamplings from Theorem \ref{thm:runningtime}.
\end{proof}

Whereas this gives an efficient way to obtain acyclic edge coloring using $16\Delta$ colors and thus matching the bound known non-constructively so far; the conjectured bound for $a(G)$ is $\Delta +2$. Alon, Sudakov and Zaks showed indeed the conjecture is true for graphs having girth $\Omega(\Delta \log{\Delta})$ \cite{alon:acyclic1}. Their algorithm can be made constructive using Beck's technique \cite{beck:lll} to obtain an acyclic edge coloring using $\Delta +2$ colors, albeit for graphs with girth significantly larger than $\Theta(\Delta \log{\Delta})$. We bridge this gap by providing constructions to achieve the same girth bound as in \cite{alon:acyclic1}, yet obtaining an acyclic edge coloring with only $\Delta+2$ colors.

The proof of Alon, Sudakov and Zaks again relies on the asymmetric LLL, but their procedure for random coloring is different from \cite{alon:acyclic, molloy:lll}. They first perform a proper coloring of the edges of $G$ using $\Delta+1$ colors \cite{vizing}. Next each edge is switched to color $\Delta+2$ with probability $1/32\Delta$. Three types of bad events are defined. In type $1$ events, two incident edges $e,f$ are colored with $(\Delta+2)$th color. Type $2$ events correspond to the case where no edge of a previously bichromatic cycle switches its coloring to the $(\Delta+2)$th color. Type $3$ events correspond to the case where a cycle with half of its edges (every other edge) having the same color after the first step, receives $(\Delta+2)$th color on half of its remaining edges resulting in a bichromatic cycle. It is sufficient to avoid these three types of events to ensure an acyclic edge coloring. It has been shown in \cite{alon:acyclic1} that by setting the values of $x$ variables to $1/512\Delta^2$, $1/128\Delta^2$ and $1/2^k\Delta^k$ for events of Type 1, 2, and 3 (with cycles of length $2k$), conditions of the asymmetric LLL (Theorem \ref{theorem:lll}) satisfy. Converting this non-constructive proof into an algorithm using our method is an easy exercise. We state the theorem below, whose proof is similar to Theorem \ref{thm:acyclic1} above, and is left as an exercise.

\begin{theorem} \label{thm:acyclic2}
There is a randomized algorithm that produces a valid acyclic coloring in expected polynomial time using $(\Delta+2)$ colors for graphs having girth $\Omega(\Delta \log{\Delta})$.
\end{theorem}

Further non-asymptotic results are known for graphs with sufficiently large girth. Muthu, Narayanan and Subramanian showed $a(G)< 6\Delta$ for graphs with girth at least $9$ and $a(G) < 4.52\Delta$ for graphs with girth at least $220$ \cite{muthua:acyclic}. Their proofs can also be made constructive using essentially the same proof of Theorem \ref{thm:acyclic1}. 
\section{Beyond the LLL Threshold}
\label{sec:beyond}
This section sketches another application of using the properties of the conditional LLL-distribution introduced in Section \ref{sec:llldistrib} in a slightly different way. While all results presented so far rely on a union bound over events in the LLL-distribution, we use here the linearity of expectation for further probabilistic analysis of events in the LLL-distribution. This already leads to new non-constructive results. Similar to the other proofs involving the LLL-distribution in this paper, this upper bound can be made constructive using Theorem \ref{thm:distrib}. Considering that the LLL-distribution approximately preserves other quantities such as higher moments, we expect that there is much more room to use more sophisticated probabilistic tools like concentration bounds to give both new non-constructive and constructive existence proofs of discrete structures with additional strong properties.

The setting we wish to concentrate on here is when a set of bad events is given from which not necessarily all but as many as possible events are to be avoided. The exemplifying application is the well known MAX-$k$-SAT problem which in contrast to $k$-SAT asks not for a satisfying assignment of a $k$-CNF formula but for an assignment that violates as few clauses as possible. Given a $k$-CNF formula with $m$ clauses a random assignment to its variables violates each clause with probability $2^{-k}$ and thus using linearity of expectation it is easy to find an assignment that violates at most $m 2^{-k}$ clauses. If on the other hand each clause shares variables with at most $2^k/e - 1$ other clauses then the LLL can be used to proof the existence of a satisfying assignment (which violates $0$ clauses) and the MT algorithm can be used to find such an assignment efficiently. But what happens when the number of clauses sharing a variables is more than $2^k/e$? Lemma \ref{lem:max-sat-core} shows that a better assignment can be constructed if it is possible to find a sparsely connected sub-formula that satisfies the LLL-condition.

\begin{lemma}\label{lem:max-sat-core}
Suppose $F$ is a $k$-CNF formula in which there exists a set of core clauses $C$ with the property that: (i) every clause in $C$ shares variables with at most $d \leq 2^k/e - 1$ clauses in $C$, and
(ii) every clause in $\overline{C}$ shares variables with at most $\gamma (2^k/e - 1)$ many clauses in $C$, for some $\gamma \geq 0$. Let $n$ and $m$ denote the total number of variables and clauses in $F$, respectively. Then, for any $\theta \geq 1/\poly(n,m)$, there is a randomized $\poly(n,m)$-time algorithm that produces, with high probability, an assignment in which all
clauses in $C$ are satisfied and at most an $(1 + \theta) 2^{-k} e^{\gamma}$ fraction of clauses from
$\overline{C}$ are violated. (If we are content with success-probability $\rho - n^{-c}$ for some constant $c$, then there is also a randomized algorithm that runs in time $\poly(n, |C|)$, satisfies all clauses in $C$, and violates at most n $(1/\rho) \cdot 2^{-k} e^{\gamma}$ fraction of clauses from $\overline{C}$. This can be useful if $|C| \ll m$.)
\end{lemma}

\begin{proof}
Briefly, the idea above is as follows. Suppose we do the obvious random
assignment to the variables: each is set to ``True'' or ``False'' uniformly
at random and independently. For any clause $C_i$, let $A_i$ be the bad
event that it is violated in such an assignment. It is well-known that we can
take $x(A_i) = e/2^k$ for all $C_i \in C$, and avoid all of these events
with positive probability: this can be made constructive using MT. Suppose
we run the MT algorithm for up to $n^c$ times its expected number of
resamplings. By Markov's inequality, the probability of MT not
terminating by then is at most $n^{-c}$. Furthermore, the probability that
at the end of this process, some clause $C_i \in \overline{C}$ is violated
can be bounded (using part (ii) of Theorem~\ref{thm:distrib}) by the following:
\[ 2^{-k} \cdot (1 - e/2^k)^{-\gamma (2^k/e - 1)} \leq e^{\gamma} \cdot 2^{-k}. \]
Thus, the expected fraction of clauses from $\overline{C}$ that are
violated in the end, is at most $2^{-k} e^{\gamma}$. Markov's inequality
and a union bound (for a sufficiently large choice of $c$)
complete the proof.
\end{proof}

Along these lines we aim to develop a general result that can be applied in cases where the number of dependencies are (slightly) beyond the LLL-threshold. For this, suppose we have a system of independent random variables $\P=\{P_1, P_2, \ldots,P_n\}$ and bad events $\A=\{A_1, A_2, \ldots, A_m\}$ with dependency graph $G = G_{\A}$ as in the introduction.  Let us consider the symmetric case in which $\Pr[A_i] \leq p$ for each $i$. Again there are only two types of constructive results known in general, in terms of only allowing a ``small'' number of the bad events to happen. It is easy to have only about $mp$ of the $A_i$ happen -- without any assumptions about $G$ -- just by using the linearity of expectation. On the other hand, if the maximum degree of $G$ is at most $1/(ep) - 1$, the conditions of
the symmetric LLL and the algorithm of \cite{moser-tardos:lll}, guarantee that we can efficiently ensure that \emph{none} of the $A_i$ happen. Interpolating between these two extremes Theorem \ref{thm:deg-more-than-thresh} characterizes the fraction $\lambda$ of events that can be avoided if the maximum degree of $G$ is by a factor of $\alpha > 1$ larger than the LLL-threshold $1/(ep) - 1$. To the best of our knowledge virtually nothing was known (even non-constructively) in this setting. Theorem \ref{thm:deg-more-than-thresh} is obtained by using the probabilistic method to construct a sparsely connected core that satisfies the LLL-conditions with a sufficiently large gap. Using the linearity of expectation in the analysis of the LLL-distribution with respect to this core the existence of a good assignment can be proven:

\begin{definition}\label{def:lambda}
For any $\alpha \geq 0$, let $\lambda(\alpha)$ be the smallest number satisfying the following:\\
For any setting of the standard form ``variables $\P$ and bad events $\A$'' in which\\
(a) the probability of any event $A \in \A$ is at most $p = o(1)$ and\\
(b) the maximum degree in the variable-sharing (dependency) graph is at most $d = \alpha (1/(ep) - 1)$,\\
there exists an assignment to variables in $\P$ that violates at most $(1 + o(1)) \lambda(\alpha) \cdot mp$ events.
\end{definition}

\begin{theorem}
\label{thm:deg-more-than-thresh}
The fraction $\lambda(\alpha)$ is upper bounded as follows:
\begin{itemize}
\item $\forall \alpha \leq 1: \lambda(\alpha) = 0$;
\item $\forall 1 < \alpha < e$:
$\lambda(\alpha) \leq e \ln(\alpha)/\alpha < 1$;
\item $\forall \alpha > e: \lambda(\alpha) = 1$.
\end{itemize}
\end{theorem}
\begin{proof}
If $\alpha \leq 1$ then the standard symmetric LLL and the MT algorithm ensure that no bad event holds. On the other hand if $\alpha > e$, then consider $m = 1/p$ events given by a single random variable $X$ which is uniformly distributed in $\{1, 2, \ldots, 1/p\}$; the $i$th event holds iff $X = i$. We have
$d = 1/p - 1$ here, and exactly one bad event holds with probability one. Thus, we cannot do better than the obvious bound of $mp$ if $\alpha > e$.

For our main case where the constant $\alpha$ satisfies $1 < \alpha < e$, we employ the probabilistic method to first determine a core subset of the bad events, and then apply Theorem~\ref{thm:distrib}. We give a proof sketch here. Since $p = o(1)$, $d = \omega(1)$. We will pick a suitable $\epsilon = o(1)$ and an appropriate constant $\beta \geq 1$; $d$ can be assumed sufficiently larger than $\beta$ since $d = \omega(1)$. Choose a random subset $\A'$ of the events $A_i$ by choosing each event independently with probability $(1 - \eps) / \alpha \beta$ and then eliminate all events from $\A$ that have more than $d/(\alpha\beta)$ neighbors in $\A'$. The Chernoff bound shows that with probability $1 - \exp(-\Omega(d\eps^2/\beta))$ (which is $1 - o(1)$ for a suitable choice of $\epsilon$ and $\beta$), the core $\A'$ has at least an $\frac{(1-\eps)^2}{\alpha \beta}$ fraction of the events, and at most an $\exp(-d\eps^2/\beta) = o(p)$ fraction of the events
get eliminated from $\A$. Therefore, there is a core $\A'$ of size at least $(1 - o(1)) m/(\alpha \beta)$ to which all events that are not eliminated (a fraction of $(1-o(p))$ events) have at most $d/(\alpha\beta)$ neighbors. If we take $x_A \sim \gamma \alpha/d$ for all $A \in \A'$ for a suitable $\gamma \in [0,1]$ then the core $\A'$ satisfies the LLL-conditions. This is the case if $\gamma$ satisfies
$$\alpha /(ed) < (\gamma \alpha /d) \cdot (1 - \gamma\alpha/d)^{d/(\beta\alpha)};$$
i.e.,
\begin{equation}
\label{eqn:cons-req}
1/e < \gamma e^{-\gamma/\beta}
\end{equation}
suffices for $d$ large enough.

Now we can apply Theorem~\ref{thm:llldistrib-nonconstr} to obtain bounds on the probability of events in the conditional LLL-distribution $D$ that avoids all events in $\A'$. It implies that a random assignment that avoids all core-events in $\A'$ makes an event $A \in (\A \setminus \A')$ that is not eliminated true with probability at most

\begin{equation}
\label{eqn:noncore-x}
\probs{_D}{A} = \prob{A} / (1 - \gamma\alpha/d)^{-d/(\alpha\beta)} \sim e^{\gamma/\beta}
\cdot p.
\end{equation}

By the linearity of expectation applied to all ``non-core'' events $A \in (\A \setminus \A')$ and using (\ref{eqn:noncore-x}), the expected total number of events
$A_i$ that happen in such an assignment is at most
\begin{equation}
\label{eqn:cons-numbad}
(1 + o(1)) \cdot mp \cdot (1 - 1/(\alpha \beta)) \cdot e^{\gamma/\beta},
\end{equation}
assuming that (\ref{eqn:cons-req}) holds and that $\epsilon = o(1)$ can be chosen suitably. From (\ref{eqn:cons-req}), we can take $1/\beta = (1 - o(1)) \cdot (1 + \ln(\gamma))/\gamma$. Plugging this into (\ref{eqn:cons-numbad}), we see that the optimal choice of $\gamma$ is $1/\alpha$. (Any choice of $\epsilon = o(1)$ that satisfies $\exp(-d\eps^2/\beta) = o(p)$, will suffice for this argument.) Substituting these choices into (\ref{eqn:cons-numbad}) yields the theorem.
\end{proof}

\begin{theorem}
\label{thm:deg-more-than-thresh-constructive}
Theorem \ref{thm:deg-more-than-thresh} can be made constructive
for any $\alpha > 0$ and any efficiently verifiable $\A$ (the verification in this case is allowed to take $\poly(n,m)$ time) that satisfies the conditions from Definition \ref{def:lambda}. That is, there is a $\poly(n,m)$-time randomized algorithm to set values for the variables in $\P$, such that the expected number of events $A_i$ that hold is at most $(1 + o(1)) \lambda(\alpha) \cdot mp$.
\end{theorem}
\begin{proof}
If $\alpha \leq 1$ then the theorem follows directly from \cite{moser-tardos:lll} and for $\alpha \geq e$ a random assignment suffices. For $1 < \alpha < e$ we make the proof from Theorem \ref{thm:deg-more-than-thresh} constructive. For it suffices to see that the success probability of the random experiment that creates the core can be made arbitrarily high by choosing $\epsilon$ and $\beta$ accordingly. This makes the probabilistic method used there directly constructive. Finally we use again our main Theorem \ref{thm:distrib} from Section \ref{sec:llldistrib} which states that the MT algorithm can be used to efficiently sample the LLL-distribution used in the proof of Theorem \ref{thm:deg-more-than-thresh}. We simply output the assignment which is produced by the MT algorithm in time $\poly(n,m)$ and the proof of Theorem \ref{thm:deg-more-than-thresh} guarantees that the expected number of violated bad events in this assignment is at most $(1 + o(1)) \lambda(\alpha) \cdot mp$ as desired.
\end{proof}

\smallskip \noindent \textbf{Remark.} Interestingly, we can use our LLL
framework instead to construct the core in the proof of
Theorem~\ref{thm:deg-more-than-thresh}. This gives a larger core than what is obtained by uniform
random core selection and thus slightly sharper
results. Briefly, the idea is as follows.
For parameters $\alpha$, $\beta$, and $\gamma$ that are similar to
those in that proof, we start with essentially the same random process for
constructing the core: for a sufficiently large constant $c_1$ (even something
slightly smaller than $3$ will suffice), we set
$\epsilon = c_1 \sqrt{(\ln d)/d}$, and place each event $A_i$ in the core
independently with probability $\theta = (1 - \epsilon)/(\alpha \beta)$.
Letting $X_i$ denote the (random) number of neighbors that $A_i$ has in
the core, note that the expected value of $X_i$ is at most
$d \theta = (1 - \epsilon)d/(\alpha \beta)$.
Now consider the system of bad events $C_i$ (one corresponding to each
original event $A_i$): $C_i$ is the event that $X_i >
d/(\alpha \beta)$. Note that each $C_i$ depends
on at most $d^2$ others; a Chernoff bound shows that for $c_1$
large enough, $\Pr[C_i] \leq 1/(e d^3)$. Thus, the LLL shows that
there exists a choice of the core that avoids all the $C_i$, with
a common value $x = x(C_i)$ for all $i$ such that
$x = \Theta(1/d^3)$; run
the Moser-Tardos algorithm on the system of bad events $C_i$ to
efficiently get a core that avoids all of the $C_i$, and let
$N_i$ be the indicator random
variable for $i$ not belonging to the core at the end of this run.
We now apply Theorem~\ref{thm:llldistrib-nonconstr} to upper-bound
the expected size $\sum_i \Pr[N_i = 1]$ of the non-core. Since
$N_i$ depends on at most $1 + d + d(d-1) = d^2 + 1$ events $C_i$,
Theorem~\ref{thm:llldistrib-nonconstr} gives
\[ \Pr[N_i = 1] \leq \frac{1 - \theta}{(1 - x)^{d^2 + 1}} =
(1 + O(1/d)) \cdot (1 - \theta). \]
Thus the expected size of the non-core is at most
$(1 + O(1/d)) \cdot mp \cdot (1 - \theta)$ after the above run of
Moser-Tardos, similar to what the alteration argument in the proof
of Theorem~\ref{thm:deg-more-than-thresh} gives. We can now
proceed (with one more run of Moser-Tardos) as in the proof of
Theorem~\ref{thm:deg-more-than-thresh}.

\bigskip
\bigskip

\noindent
\textbf{Acknowledgments:}
We thank Nikhil Bansal for helpful clarifications about the Santa Claus problem. The first author thanks Michel Goemans and Ankur Moitra for discussing the Santa Claus problem with him in an early stage of this work. Our thanks are
due to Bill Gasarch and Mohammad Salavatipour for their helpful comments.
We also acknowledge the referees for their valuable suggestions.

\bibliographystyle{abbrv}
\bibliography{af2009}

\newpage

\end{document}